\documentclass[%
 reprint,
superscriptaddress,
 amsmath,amssymb,
 aps,
]{revtex4-1}

\usepackage{amssymb}
\usepackage{amsthm}
\usepackage{comment}
\usepackage{array}
\usepackage{tikz}
\usepackage{setspace}
\usepackage{ragged2e}
\newcolumntype{R}[1]{>{\raggedleft\arraybackslash}p{#1}}

\tikzstyle{dashdotted}= [dash pattern=on 3pt off 1pt on \the\pgflinewidth off 1pt]
\tikzstyle{loosely dotted}=          [dash pattern=on \pgflinewidth off 4pt]
\usepackage{graphicx}% Include figure files
\usepackage{bm}% bold math
\newcommand{\ket}[1]{| #1 \rangle} % for Dirac bras
\newcommand{\bra}[1]{\langle #1 |} % for Dirac kets
\newcommand{\braket}[2]{\langle #1 \vphantom{#2} |  #2 \vphantom{#1} \rangle} % for Dirac brackets

\newcommand{\x}{x}

\newcommand{\e}{\mathrm{e}}
\usepackage{hyperref}% add hypertext capabilities
\usepackage{qcircuit}
\definecolor{mygreen}{RGB}{175,233,198}%inputs
\definecolor{mypurple}{RGB}{239,229,247}%algos
\definecolor{myyellow}{RGB}{255,238,170}%new output
\definecolor{mygray}{RGB}{240,238,228}
\definecolor{mydarkred}{RGB}{240,0,0}
\definecolor{mylightblue}{RGB}{176,224,230}

\definecolor{mylightred}{RGB}{255,204,153}
\definecolor{mylightorange}{RGB}{255,153,153}
\definecolor{mylightgray}{RGB}{224,224,224}

\theoremstyle{plain}
\newtheorem{thm}{\protect\theoremname}
\newtheorem{observ}[thm]{Observation}

\begin{document}

\title{Circuit-centric quantum classifiers}

\author{Maria Schuld}
\affiliation{Quantum Research Group, School of Chemistry and Physics, University of KwaZulu-Natal, Durban 4000, South Africa}
\affiliation{National Institute for Theoretical Physics, KwaZulu-Natal, Durban 4000, South Africa}
\affiliation{Quantum Architectures and Computation Group, Station Q, Microsoft Research, Redmond, WA (USA) }
\author{Alex Bocharov}%
\affiliation{Quantum Architectures and Computation Group, Station Q, Microsoft Research, Redmond, WA (USA) }
\author{Krysta Svore}%
\affiliation{Quantum Architectures and Computation Group, Station Q, Microsoft Research, Redmond, WA (USA) }
\author{Nathan Wiebe}%
\affiliation{Quantum Architectures and Computation Group, Station Q, Microsoft Research, Redmond, WA (USA) }

\date{\today}

\begin{abstract}

The current generation of quantum computing technologies call for quantum algorithms that require a limited number of qubits and quantum gates, and which are robust against errors. A suitable design approach are variational circuits where the parameters of gates are learnt, an approach that is particularly fruitful for applications in machine learning. In this paper, we propose a low-depth variational quantum algorithm for supervised learning. The input feature vectors are encoded into the amplitudes of a quantum system, and a quantum circuit of parametrised single and two-qubit gates together with a single-qubit measurement is used to classify the inputs. This circuit architecture ensures that the number of learnable parameters is poly-logarithmic in the input dimension. We propose a quantum-classical training scheme where the analytical gradients of the model can be estimated by running several slightly adapted versions of the variational circuit. We show with simulations that the circuit-centric quantum classifier performs well on standard classical benchmark datasets while requiring dramatically fewer parameters than other methods. We also evaluate sensitivity of the classification to state preparation and parameter noise, introduce a quantum version of dropout regularisation and provide a graphical representation of quantum gates as highly symmetric linear layers of a neural network.
\end{abstract}

%\pacs{Valid PACS appear here}
\maketitle

\section{\label{sec:intro} Introduction}

\setlength{\parindent}{0pt}

Quantum computing - information processing  with devices that are based on the principles of quantum theory -- is currently undergoing a transition from a purely academic discipline to an industrial technology. So called ``non-fault-tolerant'', ``small-scale'' or ``near-term'' quantum devices are being developed on a variety of hardware platforms, and offer for the first time a testbed for quantum algorithms. However, allowing for only of the order of  $1,000-10,000$ elementary operations on $50-100$ qubits \cite{mohseni17} and without the costly feature of error correction, these early devices are not yet suitable to implement the algorithms that made quantum computing famous. A new generation of quantum routines that use only very limited resources and are robust against errors has therefore been created in recent years \cite{farhi14, bremner10, huh15}. While many of those small-scale algorithms have the sole purpose of demonstrating the power of quantum computing over classical information processing \cite{harrow17}, an important goal is to find quantum solutions to useful applications.\\

One increasingly popular candidate application for near-term quantum computing is machine learning \cite{perdomo17a}. Machine learning is data-driven decision making in which a computer fits a mathematical model to data (\textit{training}) and uses the model to derive decisions (\textit{inference}). Numerous quantum algorithms for machine learning have been proposed in the past years \cite{schuld15qml, biamonte17}. A prominent strategy \cite{rebentrost14, kerenedis16, wiebe12, schuld17ibm}  is to encode data into the amplitudes of a quantum state (here referred to as \textit{amplitude encoding}), and use quantum circuits to manipulate these amplitudes. Quantum algorithms that are only polynomial in the number $n$ of qubits can perform computations on $2^n$ amplitudes. If these $2^n$ amplitudes are used to encode the data, one can therefore process data inputs in polylogarithmic time. However, most of the existing literature on amplitude encoded quantum machine learning translates known machine learning models into non-trivial quantum subroutines that lead to resource-intensive algorithms which cannot be implemented on small-scale devices. Furthermore, quantum versions of training algorithms are limited to specific, mostly convex optimisation problems. Hybrid approaches called ``variational algorithms'' \cite{mcclean16,farhi14,romero17} are much more suited to near term quantum computing and are rapidly getting popularity in the quantum research community in recent months. A general picture of variational circuits for machine learning is introduced in \cite{mitarai18}. The emphasis of low-depth circuits for quantum machine learning has been made in \cite{Verdon2017}, where the importance of entanglement as a resource has been analysed for the low-depth architectures in the context of Boltzmann machines. A very recent preprint that comes closest to the designs presented here is Farhi and Neven \cite{Farhi2018}. The latter focusses mostly on classification of discrete and discretized data that is encoded into qubits rather than amplitudes, which requires an exponentially larger number of qubits for a given input dimension. The circuit architectures proposed in the work are of more general nature compared to our focus on a slim parameter count through the systematic use of entanglement.  \\

This paper presents a quantum framework for supervised learning that makes use of the advantages of amplitude encoding, but is based on a variational approach and therefore particularly designed for small-scale quantum devices (see Fig. \ref{Fig:model_rough}). To achieve low algorithmic depth we propose a \textit{circuit-centric} design which understands a generic strongly entangling quantum circuit $U_{\theta}$ as the core of the machine learning model $f(x; \theta)=y$, where $x$ is an input, $\theta$ a set of parameters and $y$ is the prediction or output of the model. We call this circuit the \textit{model circuit}. The model circuit consists of parametrised single and controlled single qubit gates, with learnable (classical) parameters. The number of parametrised gates in the family of model circuits we propose grows only polynomially with the number of qubits, which means that our quantum machine learning algorithm has a number of parameters that is overall poly-logarithmic in the input dimension. \\

The model circuit acts on a quantum state that represents the input $x$ via amplitude encoding. To prepare such a quantum state, a static \textit{state preparation circuit} $S_x$ has to be applied to the initial ground state. After applying the state preparation as well as the model circuit, the prediction is retrieved from the measurement of a single qubit.  If the data is sufficiently low-dimensional or its structure allows for efficient approximate preparation, this yields a compact circuit that can be understood as a black box routine that executes the inference step of the machine learning algorithm on a small-scale quantum computer. \\

\begin{figure}[t]
\centering
\begin{tikzpicture}
\setlength{\baselineskip}{8pt}
\path (0,0) node[draw, fill = mypurple, rounded corners = 0.1cm, minimum height = 1cm, minimum width = 2cm, anchor=center,align = center](qd) { QPU };
\path (-2,0.25) node[anchor=east, align = center] (in){input $x$};
\path (-2,-0.25) node[anchor=east, align = center] (in2){parameters $\theta$};
\path (2,0) node[anchor=west, align = center] (out){prediction $y$};
\draw[->] (in) -- (-1,0.25);
\draw[->] (in2) -- (-1,-0.25);
\draw[->] (1,0) -- (out);
\draw(-0.8,-0.5)--(-3.3,-1.25);
\draw(0.8,-0.5)--(3.3,-1.25);
\path (0,-2.5) node[draw, fill = white!98!black, rounded corners = 0.1cm, minimum height = 2.5cm, minimum width = 7cm, anchor=center,  align = center] (out2){$
\quad $$
\Qcircuit @C=1em @R=.7em {
\lstick{\ket{0}}	&\qw & \multigate{2}{S_x} & \qw & \multigate{2}{U_{\theta}} & \qw & \meter & \cw & p(y)\\
\lstick{ \raisebox{.7em}{\vdots}}	& & \pureghost{S_x}&  & \pureghost{U_{\theta}}&\\
\lstick{\ket{0}}	&\qw & \ghost{S_x}& \qw & \ghost{U_{\theta}}& \qw \\
}$$
$};
\end{tikzpicture}
\caption{Idea of the circuit-centric quantum classifier. Inference with the model $f(x,\theta) = y$ is executed by a quantum device (the quantum processing unit or QPU) which consists of a \textit{state preparation circuit} $S_x$ encoding the input $x$ into the amplitudes of a quantum system, a \textit{model circuit} $U_{\theta}$, and a single qubit measurement. The measurement retrieves the probability of the model predicting $0$ or $1$, from which in turn the binary prediction can be inferred. The classification circuit parameters $\theta$ are learnable and can be trained by a variational scheme. }
\label{Fig:model_rough}
\end{figure}
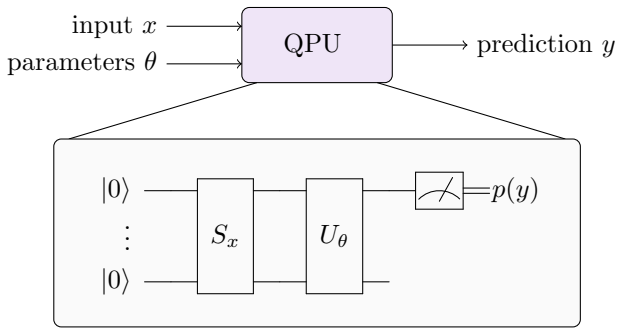

We propose a hybrid quantum-classical gradient descent training algorithm. On the analytical side we show how the exact gradients of the circuit can be retrieved from running slight variations of the inference algorithm (and for now assuming perfect precision in the prediction) a small, constant number of times and adding up the results, a strategy we call \textit{classical linear combination of unitaries}. The parameter updates are then calculated on a classical computer. Keeping the model parameters as a classical quantity allows us not only to implement a large number of iterations without worrying about growing coherence times, but also to store and reuse learnt parameters at will. Using single-batch gradient descent only requires the state preparation circuit $S_x$ to encode one input at a time. In addition to that, we can easily  improve the gradient descent scheme by standard methods such as an adaptive learning rate, regularisation and momenta.\\

We analyse the resulting \textit{circuit-centric quantum classifier} theoretically as well as via simulations to judge its performance compared to other models. We show that mathematically, a quantum circuit closely resembles a neural network architecture with unitary layers, and discuss ways to include dropout and nonlinearities. The unitarity of the ``pseudo-layers'' is a favourable property from a machine learning point of view \cite{arjovsky15, jing16}, since it maintains the length of an input vector throughout the layers and therefore circumvents notorious problems of vanishing or exploding gradients. Unitary weight matrices have also been shown to make the convergence time of gradient descent independent of the circuit depth \cite{saxe13} - an important guarantee to avoid the growing complexity of training deep architectures. Possibly the most important feature of the model is that it uses a number of parameters that is logarithmic in the data size, which is a huge saving to a neural network where the first layer already has weights at least linear in the dimension of the input. \\

In the remainder of the paper we will introduce the circuit-centric quantum classifier  in Section \ref{Sec:model}, along with design considerations for the circuit architecture in Section \ref{Sec:architecture}, as well as the training scheme in Section \ref{Sec:train}. We analyse its performance in Section \ref{Sec:sim} and show that compared with out-of-the-box methods it performs reasonable well. We finally propose a number of ways to extend the work in Section \ref{Sec:concl}.

\section{The circuit-centric quantum classifier}\label{Sec:model}

The task our model intends to solve is that of \textit{supervised pattern recognition}, and is a standard problem in machine learning with applications in image recognition, fraud detection, medical diagnosis and many other areas. To formalise the problem, let $\mathcal{X}$ be a set of inputs and $\mathcal{Y}$ a set of outputs. Given a dataset $\mathcal{D} = \{(x^1,y^1),...,(x^M,y^M)\}$ of pairs of so called \textit{training inputs} $x^m \in \mathcal{X}$ and \textit{target outputs} $y^m \in \mathcal{Y}$ for $m=1,...,M$, our goal is to predict the output $y \in \mathcal{Y}$ of a new input $x \in \mathcal{X}$. For simplicity we will assume in the following that $\mathcal{X} = \mathbb{R}^N$ and $\mathcal{Y} =\{0,1\}$, which is a binary classification task on a $N$-dimensional real input space (see Fig. \ref{Fig:problem}). Most machine learning algorithms solve this task in two steps: They first \textit{train} a model $f(x,\theta)$ with the data by adjusting a set of parameters $\theta$, and then use the trained model to \textit{infer} the prediction $y$.\\

\begin{figure}[t]
\centering
\includegraphics[width = 0.25\textwidth]{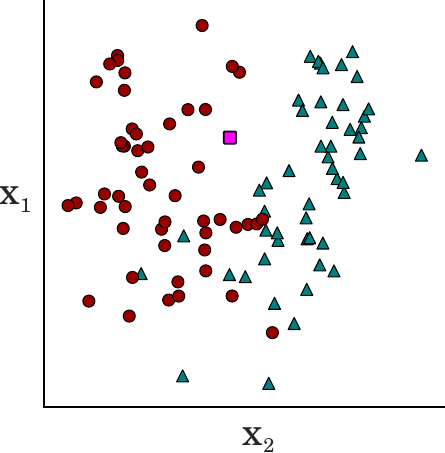}
\caption{Supervised binary classification for $2$-dimensional inputs. Given the red circle and blue triangle data points belonging to two different classes, guess the class of the new input (pink square). }
\label{Fig:problem}
\end{figure}

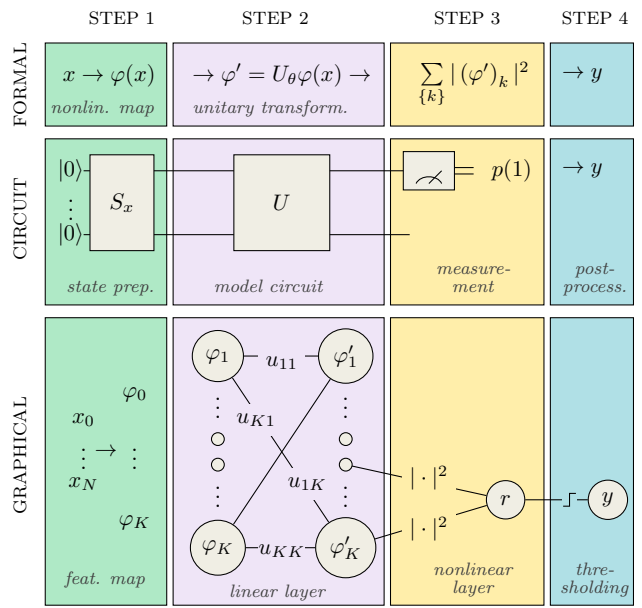
\begin{figure}[t]
\begin{tikzpicture}[every node/.style={inner sep=0,outer sep=0}, scale=0.85, every node/.style={scale=0.85}]
\setlength{\baselineskip}{8pt}
\path (0.1,3.8) node[anchor=base,align = center] {\footnotesize{STEP} \footnotesize{1 } };
\path (2.5,3.8) node[anchor=base,align = center] {\footnotesize{STEP} \footnotesize{2 } };
\path (5.5,3.8) node[anchor=base,align = center] {\footnotesize{STEP} \footnotesize{3 } };
\path (7.5,3.8) node[anchor=base,align = center] {\footnotesize{STEP} \footnotesize{4 } };

\path (-1.5,0.8) node[anchor=base,rotate=90] {\footnotesize{CIRCUIT } };
\filldraw[mygreen, draw = black] (-1.2,2) -- (0.7,2) -- (0.7,-0.6) -- (-1.2,-0.6)--cycle ;
\filldraw[mypurple, draw = black] (4.1,2) -- (0.8,2) -- (0.8,-0.6) -- (4.1,-0.6)--cycle ;
\filldraw[myyellow, draw = black] (4.2,2) -- (6.6,2) -- (6.6,-0.6) -- (4.2,-0.6)--cycle ;
\filldraw[mylightblue, draw = black] (6.7,2) -- (8,2) -- (8,-0.6) -- (6.7,-0.6)--cycle ;
\path (-0.8,1.5) node[anchor=center] (x1) {$\ket{0}$};
\path (-0.8,1) node[anchor=center] (xdots) {$\vdots$};
\path (-0.8,0.5) node[anchor=center] (xn) {$\ket{0}$};
\path (-0.1,-0.4) node[anchor=base, align = center] {\textcolor{darkgray}{\footnotesize{\textit{state prep.}}} };
\path (2.3,-0.4) node[anchor=base] {\textcolor{darkgray}{\textit{\footnotesize{model circuit}}} };
\path (5.5,-0.4) node[anchor=base, align = center] {\textcolor{darkgray}{\footnotesize{\textit{measure-}}} \\ \textcolor{darkgray}{\textit{\footnotesize{ment}}}  };
\path (7.4,-0.4) node[anchor=base, align = center] {\textcolor{darkgray}{\footnotesize{\textit{post-}}} \\ \textcolor{darkgray}{\textit{\footnotesize{process.}}}  };
\path (7.2,1.5) node[align =center] {$ \rightarrow y$};
\draw (-0.6,1.5)--(4.5,1.5) ;
\draw (-0.6,0.5)--(4.5,0.5) ;
\path (-0,1) node[align = center, fill =mygray, draw = black, minimum height = 1.5cm, minimum width = 1.cm] {$S_x$};
\path (2.5,1) node[align = center, fill =mygray, draw = black, minimum height = 1.5cm, minimum width = 1.5cm] {$U$};
\draw (4.5,1.55)--(5.5,1.55);
\draw (4.5,1.45)--(5.5,1.45);
\path (4.8,1.5) node[align = center, fill =mygray, draw = black, minimum height = 0.6cm, minimum width = 0.8cm] {};
\draw (4.6,1.3) to [bend left = 65] (5,1.3) ;
\draw (4.8,1.3) -- (5,1.5) ;
\path (6.1,1.5) node[] (xn) {$p(1)$};

\path (-1.5,2.9) node[anchor=base,rotate=90] {\footnotesize{FORMAL } };
\filldraw[mygreen, draw = black] (-1.2,3.5) -- (0.7,3.5) -- (0.7,2.2) -- (-1.2,2.2)--cycle ;
\filldraw[mypurple, draw = black] (4.1,3.5) -- (0.8,3.5) -- (0.8,2.2) -- (4.1,2.2)--cycle ;
\filldraw[myyellow, draw = black] (4.2,3.5) -- (6.6,3.5) -- (6.6,2.2) -- (4.2,2.2)--cycle ;
\filldraw[mylightblue, draw = black] (6.7,3.5) -- (8,3.5) -- (8,2.2) -- (6.7,2.2)--cycle ;
\path (-0.2,3) node[align =center] {$x \rightarrow \varphi(x)$};
\path (2.5,3) node[align =center] {$ \rightarrow \varphi'= U_{\theta}\varphi(x)  \rightarrow$};
\path (5.5,2.85) node[align =center] {$ \sum\limits_{\{k\}} |\left(\varphi'\right)_k|^2$};
\path (-0.3,2.4) node[anchor=base] {\textcolor{darkgray}{\footnotesize{\textit{nonlin. map}}} };
\path (2.4,2.4) node[anchor=base] {\textcolor{darkgray}{\textit{\footnotesize{unitary transform.}} }};
%\path (5.5,2.4) node[anchor=base] {\textcolor{darkgray}{\footnotesize{$[0,1] $} }};
%\path (7.4,2.4) node[anchor=base] {\textcolor{darkgray}{\footnotesize{$result$} }};
\path (7.2,3) node[align =center] {$ \rightarrow y$};

\path (-1.5,-2.9) node[anchor=base,rotate=90] {\footnotesize{GRAPHICAL } };
\filldraw[mygreen, draw = black] (-1.2,-0.8) -- (0.7,-0.8) -- (0.7,-5.3) -- (-1.2,-5.3)--cycle ;
\filldraw[mypurple, draw = black] (4.1,-0.8) -- (0.8,-0.8) -- (0.8,-5.3) -- (4.1,-5.3)--cycle ;
\filldraw[myyellow, draw = black] (4.2,-0.8) -- (6.6,-0.8) -- (6.6,-5.3) -- (4.2,-5.3)--cycle ;
\filldraw[mylightblue, draw = black] (6.7,-0.8) -- (8,-0.8) -- (8,-5.3) -- (6.7,-5.3)--cycle ;
\path (-0.2,-4.9) node[align =center] {\textcolor{darkgray}{\textit{\footnotesize{feat. map}}} };
\path (2.5,-5.1) node[align =center] {\textcolor{darkgray}{\footnotesize{\textit{linear layer}}} };
\path (5.5,-4.9) node[align =center] {\textcolor{darkgray}{\footnotesize{\textit{nonlinear}}} \\ \textcolor{darkgray}{\footnotesize{\textit{layer}}} };
\path (7.4,-4.9) node[align =center] {\textcolor{darkgray}{\textit{\footnotesize{thre-}}} \\ \textcolor{darkgray}{\footnotesize{\textit{sholding}}} };
\path (-0.6,-2.4) node[anchor=center] (x1) {$x_{0}$};
\path (-0.6,-2.9) node[anchor=center] (xdots) {$\vdots$};
\path (-0.6,-3.4) node[anchor=center] (xn) {$x_{N}$};
\path (-0.25,-2.9) node[anchor=center]  {$\rightarrow$};
\path (0.2,-2.) node[anchor=center] (x1) {$\varphi_{0}$};
\path (0.2,-2.9) node[anchor=center] (xdots) {$\vdots$};
\path (0.2,-4) node[anchor=center] (xn) {$\varphi_{K}$};

\path (1.5,-1.4) node[draw, anchor=center,shape=circle,  fill = mygray, minimum size=0.8cm] (p1) {$\varphi_1$};
\path (1.5,-2.1) node[anchor=center]  {$\vdots$};
\draw[fill=mygray] (1.5, -2.7) circle (0.1cm);
\draw[fill=mygray] (1.5,-3.1) circle (0.1cm);
\path (1.5,-3.5) node[anchor=center] {$\vdots$};
\path (1.5,-4.4) node[draw, anchor=center,shape=circle,  fill = mygray, minimum size=0.8cm] (pK) {$\varphi_K$};
\path (3.5,-1.4) node[draw, anchor=center,shape=circle,  fill = mygray, minimum size=0.8cm] (pd1) {$\varphi_1'$};
\path (3.5,-3.5) node[anchor=center]  {$\vdots$};
\draw[fill=mygray] (3.5, -2.7) circle (0.1cm);
\path (3.5,-2.1) node[anchor=center]  {$\vdots$};
\path (3.5,-4.4) node[draw, anchor=center,shape=circle,  fill = mygray, minimum size=0.8cm] (pdK) {$\varphi_K'$};
\draw (p1)--(pd1);
\draw (pK)--(pdK);
\draw (p1)--(pdK);
\draw (pK)--(pd1);
\path (2.5,-1.5) node[anchor=center, fill=mypurple,minimum height = 0.5cm, minimum width = 0.6cm, align=center] (u1) {$u_{11}$};
\path (2.1, -2.4 ) node[anchor=center, fill=mypurple, minimum height = 0.5cm, minimum width = 0.5cm,align=center] (ukk) {$u_{K1}$};
\path (2.9,-3.4) node[anchor=center, fill=mypurple, minimum height = 0.5cm, minimum width = 0.5cm,align=center] (u1k) {$u_{1K}$};
\path (2.5,-4.5) node[anchor=base, fill =mypurple,  align=center, inner sep=0,outer sep=0] (ukk) {$u_{KK}$};
\path (6,-3.65) node[draw,shape=circle,  fill = mygray, minimum size=0.6cm] (r) {$r$};
\draw (3.5,-3.1)--(r);
\draw[fill=mygray] (3.5,-3.1) circle (0.1cm);
\draw (pdK)--(r);
\path (4.8,-3.3) node[anchor=center, fill=myyellow, align=center] (y1) {$|\cdot|^2$};
\path (4.8,-4.1) node[anchor=center, fill=myyellow, align=center] (y2) {$|\cdot|^2$};
%\draw[->] (6.5,1.5)--(6.9,1.5);
\path (7.6,-3.65) node[draw,shape=circle,  fill = mygray, minimum size=0.6cm] (y) {$y$};
\draw (r)--(y);
\path (7,-3.75) node [minimum width =0.2cm, fill=mylightblue]  {} ;
\draw (6.9,-3.75) -- (7.0,-3.75) -- (7.0,-3.55)-- (7.1,-3.55)  ;
\end{tikzpicture}
\caption{Inference with the circuit-centric quantum classifier consists of four steps, here displayed in four colours, and can be viewed from three different perspectives, i.e. from a formal mathematical framework, a quantum circuit framework and a graphical neural network framework. In the first step, the feature map from the input space to the feature space $\mathbb{R}^N \rightarrow \mathbb{R}^K$ is executed for an input by a state preparation scheme. The quantum circuit applies a unitary transformation to the feature vector which can be understood as one linear layer (or, when decomposed into gates,  several linear layers) of a neural network. The measurement statistics of the first qubit are interpreted as the continuous output of the classifier and effectively implement a weightless nonlinear layer in which every component of the last half of all units is mapped by an absolute square and summed up. The postprocessing stage binarises the result with a thresholding function via classical computing.   }
\label{Fig:model}
\end{figure}

The main idea of the circuit-centric design is to turn a generic quantum circuit of single and $2$-qubit quantum gates into a model for classification. One can divide the full inference algorithm into four steps. As shown in Fig. \ref{Fig:model}, these four steps can be described using the language of quantum circuits, but also as a formal mathematical model, and finally, using the idea of graphical representation for neural networks, as a graphical model. \\

From a quantum circuit point of view we use the state preparation circuit $S_x$ to encode the data into the state of a $n$ qubit quantum system, which effectively maps an input $x \in \mathbb{R}^N$ to the $2^n$-dimensional amplitude vector  $\varphi(x)$ that describes the initial quantum state $\ket{\varphi(x)}$. Second, the model circuit $U_{\theta}$ is applied to the quantum state. Third, the prediction is read out from the final state $\ket{\varphi'} = U_{\theta} \ket{\varphi(x)}$ . For this purpose we measure the first of the $n$ qubits. Repeated applications of the overall circuit and measurements resolve the probability of measuring the qubit in state $1$. Lastly, the result is postprocessed by adding a learnable bias parameter $b$ and mapping the result through a step function to the output $y \in \{0,1\}$.\\

From a purely mathematical point of view, this procedure (that is, if we could perfectly resolve the probability of the first qubit by measurements) formally defines a classifier that takes decisions according to
\begin{equation}
f( x ;\theta, b) = \begin{cases} 1 \qquad \mathrm{if} \; \sum\limits^{2^n}_{k=2^{n-1}+1} \left| \left( U_{\theta} \; \varphi(x) \right)_k \right|^2 +b  > 0.5, \\
                           0  \qquad \mathrm{else}. \end{cases}.
\label{Eq:classifier}
\end{equation}
Here $\varphi: \mathbb{R}^N \rightarrow \mathbb{C}^{2^n} $ is a map that describes the procedure of information encoding via the state preparation routine ($n$ is an integer such that $2^n \geq N$), $U_{\theta} $ is the parametrised unitary matrix describing the model circuit, and $(U_\theta \varphi(x))_k$ is the $k$th entry of the result after we applied this matrix to $\varphi(x)$. The sum over the second half of the resulting vector corresponds to the single qubit measurement resulting in state $1$. Postprocessing adds the bias $b$ and thresholds to compute a binary prediction. \\

Lastly, if we formulate the four steps in the language of neural networks and their graphical representation, state preparation corresponds to a feature map on the input space, while the unitary circuit resembles a neural network of several parametrised linear layers. This is followed by two nonlinear layers, one simulating the read-out via measurement (adding the squares of some units from the previous layer) and one that maps the output to the final binary decision. We will go through the four different steps in more detail and discuss our specific design decisions for the model.

\subsection{State preparation}\label{Sec:stateprep}

There are various strategies to encode input vectors into the $n$-qubit system of a quantum computer. In the most general terms, state preparation implements a feature map $\varphi: \mathbb{R}^N \rightarrow \mathbb{C}^{2^n} $ where $n$ is the total number of qubits used to represent the features. In the following we focus on \textit{amplitude encoding}, where an input vector $x \in \mathbb{R}^N$ -- possibly with some further preprocessing to bring it into a suitable form -- is directly associated with the amplitudes of the $2^n$-dimensional `ket' vector of the quantum system written in the computational basis. This option can be extended by preparing a set of copies of the initial quantum state, which effectively implements a tensor product of copies of the input, mapping it to much higher dimensional spaces. \\

To directly associate an amplitude vector in computational basis with a data input, we require that $N$ is a power of $2$ (so that we can use all $2^n$ amplitudes of a $n$-qubit system), and that the input is normalised to unit length, $x^Tx = 1$. If $N$ is no power of $2$, we can `pad' the original input with a suitable number of zero features. (For example $x = (x_1, x_2, x_3)^T$ would be extended to $x' = (x_1,x_2, x_3, 0)^T $ ). Normalisation can pose a bigger challenge. Although many datasets carry proximity relations between vectors in their angles and not their length, some data sets can become significantly distorted by normalisation. A possible solution is to embed the data in a higher dimensional space. Practically, this can be achieved by adding non-zero padding terms before normalization. Let $N$ be the dimensionality of the original feature space, and let $c_1,\ldots, c_D$ be the padding terms that may in general depend on the informative features $x_1$ to $x_N$. The preprocessing necessary for amplitude encoding maps
\begin{equation} \label{eq:feature:normalization}
\begin{pmatrix} x_1,& ..., & x_N \end{pmatrix}^T \rightarrow  \chi \, \begin{pmatrix} x_1,&...,&  x_N, &    c_1, &..., &  c_D \end{pmatrix}^T,
\end{equation}
with
\[\chi = \frac{1}{\sqrt{\sum_j{x_j^2}+\sum_k{|c_k|^2}}},\]
on the original data.\\

For the designs investigated in this paper it is convenient to choose the padding width $D$ such that $N' = N+D$ is some exact power of $2$, and to choose  $\{c_1, \ldots,  c_D\}$ as a set of non-informative constants. This choice has in fact two desirable side-effects of feature normalisation (\ref{eq:feature:normalization}): first, is creates an `ancillary' space of dimension $D$, which in the language of neural networks is analogous to having more ``nodes'' in the first hidden layer than in the input layer; second, the constants create a state vector that is not homogeneous with respect to the vector of the original features (the importance of this will appear shortly in the discussion of the tensorial maps).\\

Preparing a quantum state whose amplitude vector in the computational basis is equivalent to the pre-processed input $\varphi(x)$ can always be done with a circuit that is linear in the number of features in the input vector, for example with the routines in Refs \cite{mottonen04,knill95,plesch11}. When more structure in the data can be exploited, preparation routines with polylogarithmic dependence on the number of features might be applicable \cite{grover02,soklakov06}. A largely uninvestigated option is also approximate state preparation of feature vectors, which may reduce the resources needed for the circuit $S_x$ at the expense of an error in the inputs. \\

To map input data into vastly higher dimensional spaces we can  apply a \textit{tensorial feature map} by preparing $d$ copies of the state \footnote{These can be exact copies when direct classical access to data is available. Otherwise these could be approximate clones of the data point generated at sufficient fidelity. }. If $\ket{\psi}$ is the `ket' vector produced by amplitude encoding, this prepares
\[ \ket{\psi} \rightarrow \underbrace{\ket{\psi} \otimes \hdots \otimes \ket{\psi} }_{d \; \mathrm{times}}.\]
For amplitude encoding with $N=2$ and $d=2$, and without any of the preprocessing described above, this would map a feature vector $(x_1,x_2)^T$ to
\[  \begin{pmatrix} x_1\\ x_2 \end{pmatrix} \otimes \begin{pmatrix} x_1\\ x_2 \end{pmatrix} = \begin{pmatrix} x_1^2 \\ x_1x_2\\ x_2x_1 \\x_2^2 \end{pmatrix}, \]
and can give rise to interesting nonlinearities that may facilitate the classification procedure in the following steps  (see also \cite{stoudenmire16}).

\subsection{The model circuit}

Given an encoded feature vector $\varphi(x) $ which is now a `ket' vector in the Hilbert space of a $n$ qubit system, the model circuit maps this ket vector to another ket vector $\varphi' = U_{\theta} \varphi(x)$ by a unitary operation $U_{\theta}$ which is parametrised by a set of variables $\theta$.

\subsubsection{Decomposition into (controlled) single qubit gates}

As described before, we decompose $U$ into
\begin{equation} \label{eq:ThaCircuit}
U = U_L \hdots U_{\ell} \hdots U_1,
\end{equation}
where each $U_{\ell}$ is either a single qubit or a two-qubit quantum gate. As a reminder, a single qubit gate $G_{k}$ acting on the $k$th of $n$ qubits can be expressed as
\begin{equation} U_l = \mathbb{I}_{0} \otimes \cdots \otimes G_{k} \otimes \cdots \otimes \mathbb{I}_{{n-1}} . \label{Eq:Ui} \end{equation}
 If the circuit depth $L$ is in $\Omega(4^n)$, this decomposition allows us to represent general unitary transformations. Remember that unitary operators are linear transformations that preserve the length of a vector, a fact that holds a number of advantages for the classifier as we will discuss later.\\

We further restrict the type of $2$-qubit gate to simplify our ``elementary parametrised gate set''. A $2$-qubit unitary gate is called \emph{imprimitive} if it can map a $2$-qubit product state into a non-product state. A common case of an imprimitive two-qubit gate is a singly-controlled single-qubit gate $C(G)$ that in standard computational basis can be written as
\begin{equation} \label{eq:controlled:unitary}
C_{a}(G_{b}) \;|x\rangle|y\rangle = |x \rangle \otimes G^{x} |y\rangle,
\end{equation}
where $G$ is a single-qubit gate other than a global phase factor on the qubit $b$ and the state $x$ of qubit $a$ is either $0$ or $1$ ($G^{0}$ is the identity). For example, $G$ could be a NOT gate, in which case the $C(G)$ is simply the frequently used CNOT gate. It is known (\cite{brylinski2002}), that single-qubit gates together with any set of imprimitive $2$-qubit gates provide for quantum universality:
\begin{observ} \label{observ:imprimitive:universality}
Circuits of the form (\ref{eq:ThaCircuit}) composed out of single-qubit gates and at least one type of imprimitive $2$-qubit gates generate the entire unitary group $U(2^n)$ in a topological sense. That is, for any $\varepsilon > 0$ and any unitary $V \in U(2^n)$ there is a circuit of the the form (\ref{eq:ThaCircuit}) the value of which is $\varepsilon$-close to $V$.
\end{observ}

To make the single qubit gates trainable we need to formulate them in terms of parameters that can be learnt. The way the parametrisation is defined can have a significant impact on training, since it defines the shape of the cost function. A single qubit gate $G$ is a $2 \times 2$ unitary, which can always be written \cite{barenco95} as
\begin{equation} \label{eq:1q:parametrization}
G (\alpha,\beta, \gamma, \phi ) = \e^{i\phi} \begin{pmatrix} \e^{i\beta} \cos \alpha &  \e^{i\gamma} \sin \alpha\\ -\e^{-i\gamma} \sin \alpha &  \e^{-i\beta} \cos \alpha \end{pmatrix}
\end{equation}
and is fully defined by four parameters $\{\alpha,\beta, \gamma, \phi\}$. For quantum gates -- where we cannot physically measure overall phase factors -- we may neglect the prefactor $\e^{i\phi}$ and only consider three learnable parameters per gate. The advantage in using angles (instead of, for example, a parametrisation with Pauli matrices) is that training does not need an additional condition on the model parameters. A disadvantage might unfavourable convergence properties of trigonometric functions close to their optima. \\

Note that there may be much more efficient ``elementary parametrised gatesets'' for a specific hardware, since some single qubit gates might naturally be parametrised in the device (i.e. where the parameter corresponds to the intensity of a laser pulse). For the agnostic case we consider here, every parametrised gate has to be decomposed into the constant elementary gate set of the physical device, which adds an efficient overhead per gate that depends on the fidelity with which we seek to approximate it (see Section \ref{Sec:performance}). \\

We treat the circuit \textit{architecture}, i.e. which qubit a certain gate acts on and where to place the controls, as fixed here and will discuss design choices in Section \ref{Sec:architecture}. Of course, strategies to learn the circuit architecture are also worth investigating, but we expect this to be a nontrivial problem due to the vast impact that each gate choice in the architecture bears for the final state (see \cite{boixo16} and the discussion on ``quantum chaos'' in random circuits).

\subsection{Read out and postprocessing} \label{subsec:readout:processing}

\begin{figure*}[t]
\begin{center}
\begin{minipage}{10cm}
\Qcircuit @C=0.4em @R=.6em {
& &  & & & B_1& & &&& &&&  & &  &B_3& &  &  & &    &  &  \\
& &  & & & && &&& &&&  & &  && &  &  & &    &  &  \\
\lstick{\ket{q_0}} & \gate{G}& \ctrl{7} & \qw &\qw &\qw &\qw &\qw &\qw & \gate{G} & \qw & \qw &\gate{G} & \ctrl{5} & \qw & \qw & \qw &\qw &\qw &\qw  &  \gate{G} &     \qw &\gate{G} & \meter \gategroup{3}{2}{10}{10}{1.em}{--}   \gategroup{3}{13}{10}{21}{1.0em}{--}   \\
\lstick{\ket{q_1}}& \gate{G}  &\qw &\qw &\qw &\qw &\qw &\qw &\gate{G} &\ctrl{-1} & \qw & \qw & \gate{G} & \qw & \qw &\qw &\qw &\gate{G} & \ctrl{5}  &\qw & \qw &  \qw\\
\lstick{\ket{q_2}} & \gate{G}&\qw &\qw &\qw &\qw &\qw &\gate{G} &\ctrl{-1} &\qw &\qw & \qw &\gate{G} & \qw & \gate{G} &  \ctrl{5}   & \qw  &\qw &\qw &\qw  &  \qw & \qw\\
\lstick{\ket{q_3}}& \gate{G} &\qw &\qw &\qw &\qw &\gate{G} &\ctrl{-1} &\qw &\qw &\qw & \qw &\gate{G} & \qw & \qw &\qw &\qw &\qw &\qw & \gate{G}& \ctrl{-3}  & \qw \\
\lstick{\ket{q_4}}& \gate{G}& \qw &\qw &\qw &\gate{G} &\ctrl{-1} &\qw &\qw &\qw &\qw & \qw &\gate{G} &\qw &\qw &\qw &\gate{G}  &  \ctrl{-3} &\qw  & \qw  & \qw &  \qw\\
\lstick{\ket{q_5}}& \gate{G}  &\qw & \qw &\gate{G} &\ctrl{-1} &\qw &\qw &\qw &\qw &\qw & \qw &\gate{G} & \gate{G}& \ctrl{-3}   & \qw & \qw &\qw & \qw &\qw   & \qw &  \qw \\
\lstick{\ket{q_6}}& \gate{G} &\qw &\gate{G} &\ctrl{-1} &\qw &\qw &\qw &\qw &\qw &\qw & \qw &\gate{G}  &\qw & \qw &\qw &\qw &\qw  &\gate{G} &\ctrl{-3}  &  \qw &  \qw\\
\lstick{\ket{q_7}}& \gate{G}  &\gate{G} &\ctrl{-1} &\qw &\qw &\qw &\qw &\qw &\qw & \qw  &  \qw &\gate{G} & \qw & \qw &\gate{G}  & \ctrl{-3}  & \qw  &\qw  &\qw  &\qw & \qw\\
}
\end{minipage}
\end{center}
\caption{Generic model circuit architecture for $8$ qubits. The circuit consists of two `code blocks' $B_1$ and $B_3$ with a range of controls of $r = 1 $ and $r=3$ respectively. The circuit consists of $17$ trainable single-qubit gates $G = G(\alpha, \beta, \gamma)$, as well as $16$ trainable controlled single qubit gates $C(G)$, which have in turn to be decomposed into the elementary constant gate set used by the quantum computer on which to implement it. If the optimisation methods are used to reduce the controlled gates to a single parameter, we have $3\cdot 33 + 1 = 100$ parameters to learn in total for this model circuit. These $100$ parameters are used to classify inputs of $2^8 = 256$ dimensions, which shows that the circuit-centric classifier is a much more compact model than a conventional feed-forward neural network. }
\label{Fig:maxent_circuit}
\end{figure*}
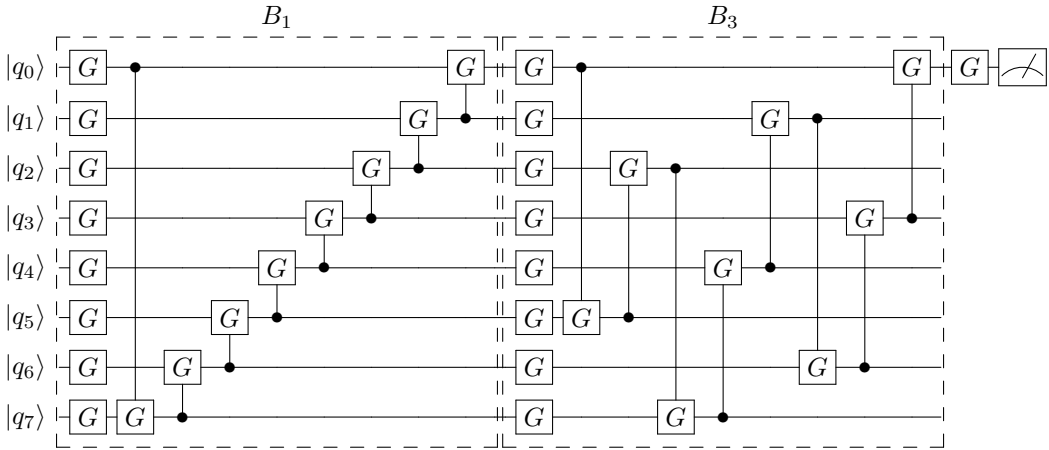

After executing the quantum circuit $U_{\theta} \varphi(x)$ in Step 2, the measurement of the first qubit (Step 3) results in state $1$  with probability\footnote{According to the laws of quantum mechanics we have to sum over the absolute square values of the amplitudes that correspond to basis states where the first qubit is in state $1$. Using the standard computational basis, this is exactly the `second half' of the amplitude vector, ranging from entry $2^{n-1}+1$ to $2^n$. }
\[  p(q_0=1, x; \theta)   = \sum\limits^{2^n}_{k=2^{n-1}+1} \left| (U_{\theta} \varphi(x))_k \right|^2 .\]
To resolve these statistics we have to run the entire circuit $S$ times and measure the first qubit. We estimate $p(q_0=1)$ from these samples $s_1,...,s_S$. This is a Bernoulli parameter estimation problem which we discuss in Section \ref{Sec:performance}.\\

The classical postprocessing (Step 4) consists of adding a learnable bias term $b$ to produce the continuous output of the model,
\begin{equation}
\pi(x;\theta, b) = p(q_0=1, x, \theta) +b.
\label{Eq:contin_out}
\end{equation}
Thresholding the value finally yields the binary output that is the overall prediction of the model:
\[ f(x ;\theta) = \begin{cases} 1 \; \mathrm{if} \; \pi(x;\theta)  > 0.5 \\
                           0  \; \mathrm{else} \end{cases}.\]
In Dirac notation the measurement result can be written as the expectation value of a $\sigma_z$ operator acting on the first qubit, measured after applying $U$ to the initial state $\ket{\varphi(x)}$.
In absence of non-linear activation, the expectation value of the $\sigma_z$ operator on the subspace of the first qubit is given by
\[ \mathbb{E} (\sigma_z) = \bra{\varphi(x)} U^{\dagger} (\sigma_z \otimes \mathbb{I} \otimes \hdots \otimes \mathbb{I} ) U \ket{\varphi(x)},  \]
and we can retrieve the continuous output via
\begin{equation} \pi(x;\theta) =  \left( \frac{\mathbb{E} (\sigma_z)}{2}    + \frac{1}{2}\right) + b. \label{Eq:dirac_model} \end{equation}

\section{Circuit architectures}\label{Sec:architecture}

Our initial goal was to build a classifier that at its core has a low-depth quantum circuit. With the circuit decomposed into $L$ single or controlled single qubit gates, we therefore want to constrain $L$ to be polynomial in $n$ which will allow us to do inference with a number of elementary quantum operations that grows only polylogarithmically in the dimension of the data set. However, this obviously comes at a price. The vectors of the form $U_{\theta } \ket{0...0}$ exhaust only a small subset of the Hilbert space of $n$ qubits. In other words, the set of amplitude vectors $\varphi' = U_{\theta} \varphi(x)$ that the circuit can `reach' is limited. In machine learning terms, this limits the flexibility of the classifier. Much like in classical machine learning, the challenge of finding a generic circuit architecture is therefore to engineer circuits (\ref{eq:ThaCircuit}) of polynomial depth that still create powerful classifiers for a subclass of datasets.

\subsection{Strongly entangling circuits}

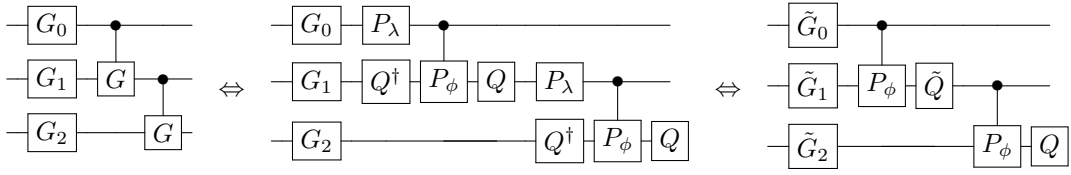
\begin{figure*}[t]
$$\Qcircuit @C=0.4em @R=.6em {
&\qw & \gate{G_0} & \qw & \ctrl{1} & \qw & \qw \\
&\qw & \gate{G_1} & \qw & \gate{G} & \ctrl{1} & \qw\\
&\qw & \gate{G_2} & \qw & \qw & \gate{G} & \qw\\
}
\quad
\raisebox{-0.85cm}{$\Leftrightarrow$}
\quad
\Qcircuit @C=0.4em @R=.6em {
&\qw & \gate{G_0} & \qw & \gate{P_{\lambda}} & \ctrl{1} &\qw & \qw & \qw & \qw & \qw \\
&\qw & \gate{G_1} & \qw &\gate{Q^{\dagger}} &\gate{P_{\phi}} &  \gate{Q} & \qw & \gate{P_{\lambda}} & \ctrl{1} & \qw\\
&\qw & \gate{G_2} & \qw & \qw & \qw & \qw \qw & \qw  & \gate{Q^{\dagger}}& \gate{P_{\phi}} & \gate{Q}\\
}
\quad
\raisebox{-0.85cm}{$\Leftrightarrow$}
\quad
\Qcircuit @C=0.4em @R=.6em {
&\qw & \gate{\tilde{G}_0} & \qw & \ctrl{1} &\qw & \qw & \qw &  \qw \\
&\qw & \gate{\tilde{G}_1} & \qw &\gate{P_{\phi}} &  \gate{\tilde{Q}} & \qw  & \ctrl{1} & \qw\\
&\qw & \gate{\tilde{G}_2} & \qw & \qw & \qw \qw & \qw  &  \gate{P_{\phi}} & \gate{Q}\\
}$$
\caption{Illustration of first step of the proof from Observation \ref{Obs:optimisation} for an example of the first $5$ gates of a codeblock of $3$ qubits with range $r=1$. Decomposing the controlled rotations and merging single qubit gates reduces the number of parameters needed to represent the model circuit architecture. For simplification the gates are displayed without indices or parameters.}
\label{Fig:optimisation}
\end{figure*}

A natural approach to the problem of circuit design is to consider circuits that prepare strongly entangled quantum states. For one, such circuits can reach `wide corners of the Hilbert space' with $U_{\theta }\,\ket{0,...,0}$. Reversibly argued, they have a better chance to project input data state $\ket{\varphi(x)}$ with the class label $y$ onto the subspace $|y\rangle \otimes |\eta\rangle$, $\eta \in \mathbb{C}^{2^{n-1}}$, which corresponds to a decision of $p(q_0) = 0,1$ in our classifier (for a zero bias). Moreover, from a theoretical point of view a classifier has to capture both short and long-range correlations in the input data, and there is mounting evidence \cite{arjovsky15, levine17} that shallow circuits may be suitable for the purpose when they are strongly entangling. \\

More specifically, we compose the circuit (\ref{eq:ThaCircuit}) out of several \emph{code blocks} $B$ (see dotted boxes in the example in Figure \ref{Fig:maxent_circuit}). A code block consists of a layer of single qubit gates $G=G(\alpha, \beta, \gamma)$ applied to each of the $n$ qubits, followed by a layer of $n/\mathrm{gcd}(n,r)$ controlled gates, where $r$ is the `range' of the control and $\mathrm{gcd}(n,r)$ is the greatest common denominator of $n$ and $r$. For $j \in [1..n/\mathrm{gcd}(n,r)]$ the $j$th $2$-qubit gate $C_{c_j}(G_{t_j})$ of a block has qubit number $t_j= (jr- r) \mod n$ as the target, qubit number $c_j = j r \mod n$ as control. A full block has the following composition,
\begin{equation} \label{eq:code:block}
B =\prod_{k=0}^{n-1} C_{c_k}(G_{t_k} ) \; \;\prod_{j=0}^{n-1} G_j .
\end{equation}
We observe that such code block is capable of entangling/unentangling all the qubits with numbers that are a multiple of $\mathrm{gcd}(n,r)$. In particular, assuming $r$ is relatively prime with $n$, all $n$ qubits can be entangled/unentangled.\\

As an example that demonstrates the entangling power of the circuit, select a block with $n=4$, $r=1$. Let all controlled gates be CNOTs and let all single qubit gates be identities, except from $G_0 = G_2 = H$, which are Hadamard gates. Applying the circuit to the basis product state $|0000\rangle$ we get the state
\[ |\psi\rangle = \frac{1}{2} (|00\rangle|00\rangle + |01\rangle|11\rangle + |10\rangle|01\rangle+|11\rangle|10\rangle).\]
If $A$ is the subsystem consisting of qubits $0,1$ and $B$ the subsystem of qubits $2,3$, then the marginal density matrix, corresponding to the state $|\psi\rangle$ and the partitioning $A \otimes B$, is completely mixed. Therefore the state $|\psi\rangle$ strongly entangles the two subsystems.

\subsection{Optimising the architecture}

The definition of the code block as per Equation (\ref{eq:code:block}) is fairly redundant. It turns out that the parameter space of the circuit (\ref{eq:code:block}) can for practical purposes be reduced to roughly $5\,n$ parameters. For this we need to introduce a controlled phase gate $C_{j}(P_{k}(\phi)), \phi \in \mathbb{R}$ that applies the phase shift $e^{i \,\phi}$ to a standard basis vector if and only if both the $j$-th and $k$-th
qubits are in state $|1\rangle$.  (Note the symmetry of the definition, which means that it does not matter which of the qubits is the control and which is the target.)

\begin{observ}\label{Obs:optimisation}
A circuit block of the form (\ref{eq:code:block}) can, up to global phase, be uniquely rewritten as
\begin{equation} \label{eq:code:block:normalized}
B = \prod_{k=0}^{n-1} R^X_k C_{c_k}(P_{t_k}) \prod_{j=0}^{n-1} G_j.
\end{equation}
where $\forall j, G_j \in SU(2)$ are single qubit gates with the usual three parameters (and, moreover, $G_j$ is an axial rotation for $j>0$), $P$ is a single-parameter phase gate, and $R^X$ is a single-parameter $X$-rotation.
\end{observ}
\begin{proof}
The proof is based on transformations of the $C_a(G_b)$ gates and subsequent normalizations of the single-qubit unitaries. Let us diagonalise the single-qubit unitary $G= Q \, D \, Q^{\dagger}$, where $Q$ is some other single-qubit unitary and $D=\mathrm{diag}(e^{i\, \lambda},e^{i\, (\lambda+\phi)})$ with  $\lambda,\phi \in \mathbb{R}$ is the diagonal matrix of the eigenvalues. Then $C_{a}(G_{b}) = Q_b C_{a}(P_b(\phi) ) Q^{\dagger}_{b} \, P_{a}(\lambda)$ (see Figure \ref{Fig:optimisation}). We further merge the $Q^{\dagger}_{b}$ with the corresponding $G_{b}$ of the code block (\ref{eq:code:block}). In case $a=0$ the $P_{a}(\lambda)$ can be commuted to the beginning of the layer and merged with $G_0$. In case $a \neq 0$ the $P_{a}(\lambda)$ can be commuted through the end of the layer and either merged into the next layer or, if we are looking at the last layer in the circuit, traced out. At this point the action of the circuit (\ref{eq:code:block}) is equivalent to that of
$(\prod_{k=0}^{n-1} \tilde{Q}_{t_k} \, C_{c_k}(P_{t_k}) ) \prod_{j=0}^{n-1} \tilde{G}_j$ where $\tilde{G}_j$ and $\tilde{Q}_{t_k}$ are updated single qubit gates from the merging operation. Note that the single qubit gates in the following layer are also updated.\\

In the second round of the we split all the single-qubit gates $\tilde{Q}$ up to global phase into product of three rotations $\tilde{Q} \sim R_Z(\mu_1) R_X(\mu_2) R_Z((\mu_3)$.
We conclude the proof by noting that each of the diagonal operators $R_Z(mu_1)_{t_j}, R_Z((\mu_3)_{t_j}$ can be commuted through all controlled phase gates to either the end of the layer or to the beginning (in which case it can be merged with one of the $\tilde{G}_j$ gates).
\end{proof}

To summarize, with the possible exception of the last layer in the classifier, a layer is described (up to a global phase) by at most $5\,n$ parameters,
at most $n$ for all controlled phase gates $C(P)$, at most $n$ for all $x$-rotations $R^X$ and at most $3\,n$ for all fully parametrised single qubit gates $G$.

\subsection{Graphical representation of gates}\label{Sec:nn} % WAS: Equivalence to neural networks}\label{Sec:nn}

As a product of elementary gates, the model circuit $U_x$ can be understood as a sequence of linear layers of a neural network with the same number of units in each ``hidden layer''. This perspective facilitates the comparison of the circuit-centric quantum classifier with widely studied neural network models, and visualises the connectivity power of (controlled) single qubit gates. The position of the qubit (as well as the control) determine the architecture of each layer, i.e. which units are connected and which ``weights'' are tied in a ``gate-layer''. \\

To show an example, consider a Hilbert space of dimension $2^n$ with $n=2$ qubits $\ket{q_0 q_1}$. A single qubit unitary $G$ applied to $q_0$ would have the following matrix representation
\[ G_{0} = \begin{pmatrix} \def\arraystretch{0.4}
 \e^{i\beta} \cos \alpha & 0      & \e^{i\gamma} \sin \alpha & 0   \\
0      &  \e^{i\beta} \cos \alpha & 0      & \e^{i\gamma} \sin \alpha \\
-\e^{-i\gamma} \sin \alpha & 0      &  \e^{-i\beta} \cos \alpha & 0     \\
0      & -\e^{-i\gamma} \sin \alpha & 0      &  \e^{-i\beta} \cos \alpha
\end{pmatrix},\]
 while the same unitary but controlled by qubit $q_1$, $C_{1}(G_{0})$ has matrix representation
 \[C_{1}(G_{0}) = \begin{pmatrix}  \def\arraystretch{0.4}
1 & 0      & 0 & 0      \\
0      &  \e^{i\beta} \cos \alpha & 0      & \e^{i\gamma} \sin \alpha  \\
0 & 0      & 1 & 0     \\
0      &-\e^{-i\gamma} \sin \alpha & 0      &  \e^{-i\beta} \cos \alpha  \\
\end{pmatrix}\]
At the same time, these two gates can be understood as layers with connections displayed in Figure \ref{Fig:graph}.\\

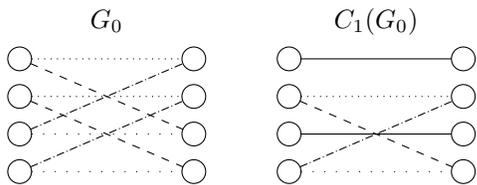
\begin{figure}[t]
\begin{center}
\begin{tikzpicture}[every node/.style={inner sep=0,outer sep=0}]
\path (1,0.5) node[] {$G_{0}$};
\path (0,0) node[draw, shape=circle,anchor = east] (i0) {\phantom{$h$}};
\path (0,-0.5) node[draw, shape=circle,anchor = east](i1) {\phantom{$h$}};
\path (0,-1) node[draw, shape=circle,anchor = east] (i2) {\phantom{$h$}};
\path (0,-1.5) node[draw, shape=circle,anchor = east] (i3) {\phantom{$h$}};
\path (2,0) node[draw, shape=circle,anchor = west] (o0) {\phantom{$h$}};
\path (2,-0.5) node[draw, shape=circle,anchor = west](o1) {\phantom{$h$}};
\path (2,-1) node[draw, shape=circle,anchor = west] (o2) {\phantom{$h$}};
\path (2,-1.5) node[draw, shape=circle,anchor = west] (o3) {\phantom{$h$}};
\draw[dotted] (i0) -- (o0);
\draw[dashed] (i0) -- (o2);
\draw[dotted] (i1) -- (o1);
\draw[dashed] (i1) -- (o3);
\draw[dashdotted] (i2) -- (o0);
\draw[loosely dotted] (i2) -- (o2) node[] {};
\draw[dashdotted] (i3) -- (o1) ;
\draw[loosely dotted] (i3) -- (o3);

\end{tikzpicture} ~~~~~~
\begin{tikzpicture}[every node/.style={inner sep=0,outer sep=0}]
\path (5,0.5) node[] {$C_{1}(G_{0})$};
\path (4,0) node[draw, shape=circle,anchor = east] (ii0) {\phantom{$h$}};
\path (4,-0.5) node[draw, shape=circle,anchor = east](ii1) {\phantom{$h$}};
\path (4,-1) node[draw, shape=circle,anchor = east] (ii2) {\phantom{$h$}};
\path (4,-1.5) node[draw, shape=circle,anchor = east] (ii3) {\phantom{$h$}};

\path (6,0) node[draw, shape=circle,anchor = west] (oo0) {\phantom{$h$}};
\path (6,-0.5) node[draw, shape=circle,anchor = west](oo1) {\phantom{$h$}};
\path (6,-1) node[draw, shape=circle,anchor = west] (oo2) {\phantom{$h$}};
\path (6,-1.5) node[draw, shape=circle,anchor = west] (oo3) {\phantom{$h$}};
\draw (ii0) -- (oo0);
\draw[dotted] (ii1) -- (oo1);
\draw[dashed] (ii1) -- (oo3);
\draw (ii2) -- (oo2) ;
\draw[dashdotted] (ii3) -- (oo1) ;
\draw[loosely dotted] (ii3) -- (oo3);

\end{tikzpicture}
\end{center}
\caption{Graphical representation of quantum gates. Left: A single qubit gate applied to the first qubit of a $2$-qubit register. Right: A single qubit gate and a controlled single cubit gate applied to a two-qubit register. A solid line corresponds to a unit weight, while the other lines stand for a variable weight parameter. The same line styles indicate the same weights.}
\label{Fig:graph}
\end{figure}
 
It becomes obvious that a single qubit gate connects two sets of two variables with the same weights, in other words, it ties the parameters of these connections. The control removes half of the ties and replaces them with identities. A quantum circuit can therefore be understood as an analog of a neural network architecture with highly symmetric, unitary linear layers, and controls break some of the symmetry. Note that although we speak of linear layers here, the weights (i.e., the entries of the weight matrix representing a gate) have a nonlinear dependency on the model parameters $\theta$, a circumstance that plays a role for the convergence of the hybrid training method.

%\textcolor{red}{[there was the "Nonlinearities" subsec here. Restore it eventually]}

\section{Training}\label{Sec:train}

We consider a stochastic gradient descent method for training. The parameters that define every single qubit gate of the quantum circuit are at every stage of the quantum algorithm classical values. However, we are computing the model function on a quantum device, and have therefore no `classical' access to its gradients. This means that the training procedure has to be a hybrid scheme that combines classical processing to update the parameters, and quantum information processing to extract the gradients. We will show how to use the quantum circuit to extract estimates of the analytical gradients, as opposed to other proposals for variational algorithms based on derivative-free or finite-difference gradients (see \cite{guerreschi17}). A related approach, but for a different gate representation, has been proposed during the time of writing in Ref. \cite{Farhi2018}.

\subsection{Cost function}

We choose a standard least-squares objective to evaluate the cost of a parameter configuration $\theta$ and a bias $b$ given a training set, $\mathcal{D} = \{ (x^1,y^1),...,(x^M,y^M)\}$,
\[ C(\theta, b; \; \mathcal{D} ) =  \frac{1}{2} \sum\limits_{m=1}^{M}|\pi(x^m; \; \theta, b) - y^m|^2,\]
where $\pi$ is the continuous output of the model defined in Equation (\ref{Eq:contin_out}). Note that we can easily add a regularisation term (i.e., an $L_1$ or $L_2$ regulariser) to this objective, since it does not require any additional quantum information processing. For the sake of simplicity we do not consider regularisation in this paper. \\

Gradient descent updates each parameter $\mu$ from the set of circuit parameters $\theta$ via
\[\mu^{(t)} = \mu^{(t-1)} - \eta \frac{\mu C(\theta, b; \; \mathcal{D} )}{\partial \theta },\]
and similarly for the bias,
\[b^{(t)} = b^{(t-1)} - \eta \frac{\partial C(\theta, b; \; \mathcal{D} )}{\partial b}.\]
The learning rate $\eta$ can be adapted during training and we can also add momenta to the updates, which can significantly decrease the convergence time.  \\

In \textit{stochastic} gradient descent, we do not consider the entire training set $\mathcal{D}$ in every iteration, but only a subset or \textit{batch} $\mathcal{B} \subset \mathcal{D}$ \footnote{Stochastic gradient descent originally referred to the case $B=1$, but is often used as a synonym for \textit{minibatch} gradient descent as opposed to \textit{batch} gradient descent using the full dataset.}.  The derivatives in the parameter updates are therefore taken with respect to $C(\theta, b; \; \mathcal{B} )$ instead of $C(\theta, b; \; \mathcal{D} )$. In principle, quantum computing allows us to encode a batch of $B$ training inputs into a quantum state in superposition and feed it into the classifier, which can be used to extract gradients for the updates from the quantum device. However, guided by the design principle of a low-depth circuit, this would extend the state preparation routine to be in $\mathcal{O}(BN)$ for general cases, where $N$ is the size of each input in the batch (which becomes even worse for more sophisticated feature maps in Step 1). We therefore consider single-batch gradient descent here (i.e., $B=1$), where only one randomly sampled training input is considered in each iteration. Single-batch stochastic gradient descent can have favourable convergence properties, for example in cases where there is a lot of data available \cite{bottou10}.

\subsection{Hybrid gradient descent scheme}

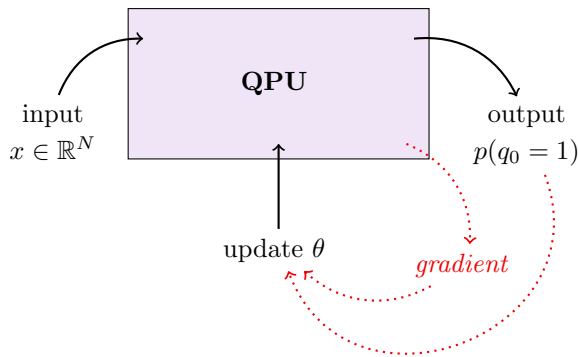
\begin{figure}[t]
\begin{tikzpicture}
\setlength{\baselineskip}{14pt}
\filldraw[mypurple, draw = black] (-2,1) -- (2,1) -- (2,-1) -- (-2,-1)--cycle ;
\path (0,0) node[rounded corners=0.1cm, anchor=center, align = center] {\textbf{QPU}  };
\path (-3,-1) node[anchor=base, align = center] (in) {input \\ $x \in \mathbb{R}^N$ };
\path (3.3,-1) node[anchor=base, align = center] (out) {output \\$ p(q_0 =1) $ };
\draw[->, thick] (in) to [bend left=35] (-1.8,0.6);
\draw[->, bend left, thick] (1.8,0.6)to [bend left=35] (out);
\path (0,-2.3) node[anchor=base, align = center] (pars) {update $\theta$ };
\draw[->,  thick] (pars) --(0,-0.8);

\path (2.5,-2.5) node[ mydarkred, anchor=base, align = center] (delta) {\textit{gradient}  };
\draw[->, mydarkred, dotted, bend left, thick] (1.7,-0.8) to [bend left=35] (delta);
\draw[->,  mydarkred, dotted, bend left, thick] (delta) to [bend left=35]  (pars);
\draw[->,  mydarkred, dotted, bend left, thick] (out) to [bend left = 90 , looseness = 1.6]  (pars);

\end{tikzpicture}
\caption{Idea of the hybrid training method. The quantum processing unit (QPU) is used to compute outputs and gradients of the model in order to update the parameters for each step of the gradient descent training.   }
\label{Fig:scheme}
\end{figure}

The derivative of the objective function with respect to a model parameter $\nu = b,\mu $ (where $\mu \in \theta$ is a circuit parameter) for a single data sample $\{(x^m, y^m)\}$ is calculated as
\[\frac{\partial C}{\partial \nu}  = \left( \pi(x^m;\nu)  - y^m \right) \; \partial_{\nu} \pi(x^m;\nu) . \]
Note that $\pi(x^m;\nu)$ is a real-valued function and the $y^m$ and the parameters are also real-valued. Hence $\frac{\partial C}{\partial \nu} \in \mathbb{R}$. \\

While $\pi(x^m;\nu)$ is a simple prediction we can get from the quantum device, and $y^m$ is a target from the classical training set, we have to look closer at how to compute the gradient $\partial_{\nu} \pi$.  For $\nu=b$ this is in fact trivial, since
\[ \partial_{b} \pi(x^m;b) = 1. \]
In case of $\nu = \mu$, the gradient forces us to compute derivatives of the unitary operator. In the following we will calculate the gradients in vector as well as in Dirac notation and show how a trick allows us to estimate these gradients using a slight variation of the model circuit $S_x$.\\

The derivative of the continuous output of the model with respect to the circuit parameter $\mu$ is formally given by
\begin{eqnarray}
 \partial_{\mu} \pi(x^m;\mu)  &=& \partial_{\mu} \;  p(q_0=1; x^m, \theta) \nonumber \\
&=& \partial_{\mu} \sum\limits^{2^n}_{k=2^{n-1}+1} \left( U_{\theta} \; \varphi(x) \right)^{\dagger}_k \left( U_{\theta} \; \varphi(x) \right)_k \nonumber \\
&=& 2 \mathrm{Re}\left\{  \sum\limits^{2^n}_{k=2^{n-1}+1} \left( \partial_{\mu} U_{\theta} \; \varphi(x) \right)^{\dagger}_k \left( U_{\theta} \; \varphi(x) \right)_k  \right\}. \nonumber
\end{eqnarray}
The last expression contains the `derivative of the circuit', $\partial_{\mu} U_{\theta}$, which is given by
\begin{equation*}
\partial_{\mu} U_{\theta} = U_L \hdots (\partial_{\mu} U_i) \hdots U_1,
\end{equation*}
where we assume for simplicity that only the parametrised gate $U_i$ depends on parameter $\mu$. If the parameters of different unitary matrices are tied then the derivative can simply be found by applying the product rule.\\

In Dirac notation, we have expressed the probability of measuring the first qubit in state $1$ through the expectation value of a $\sigma_z$ operator acting on the same qubit,  $p(q_0=1; x^m, \theta) = \frac{1}{2} ( \bra{U_{\theta} \varphi(x)} \sigma_z \ket{U_{\theta} \varphi(x)} +1) $ (see Equation \ref{Eq:dirac_model}). We can use this expression to write the gradient in Dirac notation,
\begin{equation}
  \partial_{\mu} \;  \pi(x^m;\theta,b)    =  \mathrm{Re} \{ \bra{ (\partial_{\mu}U_{\theta})\varphi(x^m)} \sigma_z \ket{U_{\theta} \varphi(x^m)} \} . \label{Eq:dirac_deriv}
  \end{equation}
This notation reveals the challenge in computing the gradients using the quantum device. The gradient of a unitary is not necessarily a unitary, which means that $\ket{ (\partial_{\mu}U_{\theta})\varphi(x^m)}$ is not a quantum state that can arise from a quantum evolution. How can we still estimate gradients using the quantum device?

\subsection{Classical linear combinations of unitaries}

It turns out that in our architecture we can always represent  $\partial_{\mu} U_{\theta}$ as a linear combination of unitaries. Linear combination of unitaries is a known technique in quantum mechanics \cite{childs12}, where the sum is implemented in a coherent fashion. In our case where we allow for classical postprocessing, we do not have to apply unitaries in superposition, but can simply run the quantum circuit several times and collect the output. This is what we will call \textit{classical linear combinations of unitaries} here.\\

Consider the derivative of $U_i$ for the single-qubit gate defined in Equation (\ref{Eq:Ui}),
\[ \partial_{\mu}  U_i  = \mathbb{I} \otimes \cdots \otimes  \partial_{\mu}  G(\alpha, \beta, \gamma)  \otimes \cdots \otimes \mathbb{I}, \]
where $G(\alpha, \beta, \gamma)$ is given in the parametrisation introduced in Equation (\ref{eq:1q:parametrization}) and discounting the global phase.  The derivatives of the single qubit gate $G(\alpha, \beta, \gamma)$ with respect to the parameters  $\mu = \alpha, \beta,\gamma$ are as follows:
\begin{eqnarray}
 \partial_{\alpha} G  &=&  G(\alpha + \frac{\pi}{2}, \beta, \gamma )  \label{Eq:Gderiv_alpha} \\
 \partial_{\beta} G &=& \frac{1}{2} G(\alpha, \beta + \frac{\pi}{2} , 0  ) + \frac{1}{2}  G(\alpha, \beta + \frac{\pi}{2} , \pi ) \label{Eq:Gderiv_beta1} \\
\partial_{\gamma} G &=& \frac{1}{2}G(\alpha, 0, \gamma + \frac{\pi}{2}  ) + \frac{1}{2}  G(\alpha, \pi, \gamma + \frac{\pi}{2} ) \label{Eq:Gderiv_beta2}
\end{eqnarray}
One can see that while the derivative with respect to $\alpha$ requires us to implement the same gate but with the first parameter shifted by $\frac{\pi}{2}$, the derivative with respect to $\mu=\beta$ [$\mu=\gamma$] is a linear combination of single qubit gates where the original parameter $\beta$ [$\gamma$] is shifted by $\frac{\pi}{2}$, while $\gamma$ [$\beta$] is replaced by $0$ or $\pi$.\\

Differentiating a controlled single qubit gate is not that immediate, but fortunately we have
\[  \partial_{\mu } \; C(G)  = \frac{1}{2} \bigg(  C (\partial_{\mu} G) - C ( - \partial_{\mu} G)  \bigg), \]
which means that the derivative of the controlled single qubit gate is half of the difference between a controlled derivative gate and the controlled negative version of that gate. In our design, when $\mu=\alpha$, each of the two controlled gates is unitary, while $\mu = \beta, \gamma$ requires us to use the linear combinations in (\ref{Eq:Gderiv_beta1}) and (\ref{Eq:Gderiv_beta2}).\\

If we plug the gate derivatives back into the expressions for the gradient in Equation (\ref{Eq:dirac_deriv}), we see that the gradients, irrespective of the gate or parameter, can be computed as `classical' linear combinations of the form
\begin{equation*} \partial_{\mu} \;  \pi(x^m;\theta,b)   =  \sum_{j=1}^J a_j \;  \mathrm{Re} \left\{ \bra{ U_{\theta^{[j]}}  \varphi(x^m)} \sigma_z \ket{U_{\theta}\varphi(x^m)} \right\},  \end{equation*}
where $\theta^{[j]}$ is a modified vector of parameters corresponding to a term appearing in Equations (\ref{Eq:Gderiv_alpha} - \ref{Eq:Gderiv_beta2}), and $a_j$ is the corresponding coefficient also stemming from the Equations (\ref{Eq:Gderiv_alpha} - \ref{Eq:Gderiv_beta2}).
If there is no parameter tying between the constituent gates, for example, then $J$ is either $2$ or $4$ depending on whether the gate containing parameter $\mu$ is a one- or two-qubit gate. For each circuit, the eventual derivative has to be estimated by repeated measurements, and we will discuss the number of repetitions in the following section.  \\

The last thing to show is that we can compute the terms $\mathrm{Re} \{ \bra{ U_{\theta^{[j]}}  \varphi(x^m)} \sigma_z \ket{U_{\theta}\varphi(x^m)} \}$ with the quantum device, so that classical multiplication and summation can deliver estimates of the desired gradients. \\

\begin{observ}
Given two unitary quantum circuits $A$ and $B$ that act on a $n$ qubit register to prepare the two quantum states $\ket{A}, \ket{B}$, and which can be applied conditioned on the state of an ancilla qubit, we can use the quantum device to sample from the probability distribution $p = \frac{1}{2} + \frac{1}{2} \; \mathrm{Re} \braket{A}{B}$
\end{observ}

\begin{proof}
The proof of this observation follows from the exact same reasoning that underlies the Hadamard test.  For concreteness, we specify the algorithm below.
 We use the circuits $A,B$ to prepare the two states $\ket{A}, \ket{B}$ conditioned on an ancilla,
 \[\frac{1}{\sqrt{2}} \left(\ket{0} \ket{A} + \ket{1} \ket{B} \right).\]
 Applying a Hadamard on the ancilla yields
 \[\frac{1}{2} \left(\ket{0} (\ket{A} + \ket{B}) + \ket{1} (\ket{A} - \ket{B}) \right),
 \]
 where the probability of the ancilla to be in state $0$ is given by
\[ p(a=0) = \frac{1}{2} + \frac{1}{2} \mathrm{Re} \braket{A}{B}.\]
\end{proof}
To use this interference routine we have to add an extra qubit and implement $U_{\theta}$ and $U_{\theta^{[j]}}$ conditioned on the state of the ancilla. Since these two circuits coincide in all except from one gate, we do in fact only need to apply the differing gate in conditional mode. This turns a single qubit gate into a singly controlled single qubit gate, and a controlled gate into a double controlled gate.  The desired value  $\mathrm{Re} \braket{A}{B}$ can be derived by resolving $p(a=0)$ through measurements and computing
\[   \mathrm{Re} \braket{A}{B} =2p(a=0) - 1. \]

\subsection{Dropout}

%\textcolor{red}{[TODO ITEM 3: Explain dropout.]}[alexeib, 12/17/17]
Despite the relatively small parameter space, our circuit-centric architecture is not immune to overfitting. Benchmarking on smaller data sets reveals cases where the training data is fit perfectly (zero misclassifications) by a model with exponentially few parameters, but the same model has significant generalization errors on the test holdout.\\

The approach that often helps is a simple \emph{dropout regularization} that is both quantum-inspired and quantum ready (in the sense that it is easy in both classical simulation and quantum execution). The essence of the approach is to randomly select and measure one of the qubits, and set it aside for a certain number $N_{dropout}$ of parameter update epochs. After that, the qubit is re-added to the circuit and another qubit (or, perhaps, no qubit) is randomly dropped.
This strategy works by ``smoothing'' the model fit and it generally inflates the training error, but often deflates the generalization error.\\

The effect of such dropout regularization is similar, in spirit, to dropout regularization in a traditional neural network when the dropout probability $p=0.5$ is used. Indeed, freezing a randomly chosen qubit for a certain number of epochs prevents a half of the amplitudes in the amplitude encoding from affecting the stochastic gradient during these epochs. In the graphical representation of the circuit-centric classifier this is analogous to removing a half of the nodes from a hidden layer for a certain number of epochs.

\subsection{Performance analysis} \label{Sec:performance}
In order to use the circuit-centric quantum classifier with near-term quantum devices, we need to motivate that it only requires a small number of qubits, a low circuit depth as well as a high error tolerance. After introducing the details of the algorithms for inference and training, we want to discuss these three points in more detail.

\subsubsection{Circuit depth and width}
The number of qubits needed for the circuit-centric quantum classifier (if we use amplitude encoding as explained in Section \ref{Sec:stateprep}) is given by $d \lceil \log_2 N \rceil $ where $N$ is the dimension of the inputs and $d$ is the number of copies we consider for a tensorial feature map. For example, if $d=1$, we can process a dataset of $1000$-dimensional inputs with $ n =10$ qubits. With about $50$ qubit we
can use a tensorial feature map of $d=5$ (i.e., prepare $5$ copies of the state) and map the data into a $2^{50}$ dimensional feature space. For the inner products subroutine in the hybrid training scheme, we need one extra ancilla qubit. The algorithm is therefore very compact as much as circuit width is concerned, a feature stemming from the amplitude encoding strategy.\\

The bottleneck of the circuit depth is the state preparation routine $S_x$. Comparably, implementing the model circuit costs a negligible amount of resources. Using an architecture with $K$ codeblocks of ranges $(r_1,...,r_K)$ and $n$ qubits, we need
\[  Kn + \sum\limits_{k=1}^K n/\mathrm{gcd}(n,r_k)\]
parametrised (controlled) single qubit gates to implement $U_{\theta}$, which is polynomial in the number of qubits. Each of these gates has to be decomposed into the elementary constant gate set used in the physical implementation of the quantum computer. Every parametrised single qubit gate can be efficiently translated into circuit $\tilde{G}$ of at most $\mathcal{O}(\log \frac{1}{\delta})$ constant elementary gates from a given gate set such as ``Clifford-plus-T'' to a fidelity of at least $(1- \delta)$ (cf. \cite{bocharov2013,kliuchnikov2014,selinger2015,kliuchnikov2015}).
 %P. Selinger. Efficient Clifford+T approximation of single-qubit operators. QIC, 15(1-2):159–180, Jan. 2015
 %A. Bocharov, Y. Gurevich, and K. M. Svore. Efficient decomposition of single-qubit gates into V basis circuits. Phys. Rev. A, 88:012313 (13 pages), 2013
 % A Framework for Approximating Qubit Unitaries, Vadym Kliuchnikov, Alex Bocharov, Martin Roetteler, Jon Yard, https://arxiv.org/abs/1510.03888
 %Asymptotically Optimal Topological Quantum Compiling Vadym Kliuchnikov, Alex Bocharov, Krysta M. Svore, Phys. Rev. Lett. 112, 140504 (2014)
 Methods such as automated optimization \cite{nam17} may reduce the costs further.\\

General state preparation can in the worst case require  $c_{\mathrm{cn}} 2^n$ CNOT gates as well as $c_{\mathrm{sgl}} 2^n$ single qubit gates. For current algorithms $c_{\mathrm{sgl}}$ and $c_{\mathrm{cn}}$ is equal to or slighly larger than $1$ \cite{knill95,mottonen04,vartiainen04,plesch11,iten16}. This means that for the example of $N=1000$ from above, we would indeed require $2 \cdot 2^n = 2048$ gates only to prepare the states. Issues of fidelity arise, since without error correction we cannot guarantee to prepare a close enough approximation of $x$. Our simulations show that adding $5\%$ noise to the inputs does not change the classification result significantly, which suggests that the classifier is rather robust against input noise. Still, until error correction becomes a reality, it is therefore advisable to focus on lower-dimensional  datasets. Two interesting exceptions have to be mentioned. First, if an algorithm is known that efficiently allows us to approximate the (preprocessed) inputs with a product state, $x \approx a_1 \otimes \cdots \otimes a_K$ the resources reduce to the number of gates required to prepare the $a_1,...,a_K$ in amplitude encoding \cite{grover02}. Second, as other authors in quantum machine learning research, we point out that if the data is given by a shallow and robust digital quantum simulation routine performed on the same register of qubits, our classifier can be used to train with `quantum data', or inputs that are `true' wavefunctions.

\subsubsection{Number of repetitions for output estimation}
The continuous output of the circuit-centric quantum classifier was based on the probability of measuring the first qubit in state $1$. To resolve this number, we have to repeat the entire algorithm multiple times. Each measurement samples from the Bernoulli distribution
$p(q_0=1) = \nu  $,
and we want to estimate $\nu$ from the $S$ samples $q^1_1,...,q^S_1$. The number of samples needed to estimate $\nu$ at error $\epsilon$ with probability $>2/3$ scales as $O(\mathrm{Var}(\sigma_z)/\epsilon^2))$, where $\mathrm{Var}(\sigma_z)$ is the variance of the sigma-z operator that we measure with respect to the final quantum state \cite{guerreschi17, schuld17ibm}. If amplitude estimation is used then the number of repetitions of circuit centric classifier falls into $O(1/\epsilon)$ at a price of increasing the circuit depth by a factor of $O(1/\epsilon)$.

%Setting the variance for now to $1$, if we want to estimate $\nu$ with an error of $0.01$, we need more than $10,000$ samples.\\ TODO:REVIEW: deprecated per Nathan.

\subsubsection{Parameter noise}

An important feature of the circuit-centric classifier is its robustness to noise in the inputs and parameters.
Suppose $\delta>0$ is some small value and we allow parameter permutations (resp. input permutations) such that for each constituent gate $G$ the permuted gate $G'$ is $\delta$-close to $G$: $||G-G'||<\delta$. (Or for encoded input $\phi$ a perturbed input $\phi'$ is $\delta$-close to $\phi$.) we allow certain imprecisions in some or all parameter values and that such imprecisions are bounded below some constant $\delta$. %[TODO: SENTENCE LONG AND NOT CLEAR]
Since all the constituent operations are unitary, the impact of the parameter imprecisions is never amplified across the circuit at the defect imposed by the imperfect circuit on the final state before the measurement is bounded by
$4\, L \, \delta$ in the worst theoretical case, where $L$ is the number of elementary parametrised gates which have at most $4$ parameters. In practice the propagated error should be much smaller that this bound.\\

The same analysis applies to imperfections in the quantum gates execution (other than parameter drift). There is no amplification of defect across the circuit and the imperfection of the final state is bounded by the sum of imperfections of individual gates. Finally, the ket encoding of the input data does not have to be perfect either. A possible imperfection or approximation during the state preparation will not be amplified by the classification circuit, and the drift of the pre-measurement state will be never be greater than the drift of the initial state. Another widely advertised advantage of variational quantum algorithms is that they can learn to counterbalance systematic errors in the device architecture -- for example when one gate always over-rotates the state by the same value.\\

\begin{table}[t]
\def\arraystretch{1.5}
\begin{tabular}{R{1.7cm} R{1.6cm} R{1.8cm} R{1.5cm}  R{1.3cm}   }
 \toprule
\footnotesize Noise level &\footnotesize RI Mean  & \footnotesize RI St.Dev. &\footnotesize  RI max &\footnotesize RI min \\
 \hline
0.1\% &1\%&1.47\% & 3.5\% & 0\%  \\
1\% &7\%& 6.67\% & 17\%&0\%  \\
10\% &60.2\% & 55.8\% & 192.3\% & 0\% \\
\botrule
\end{tabular}
\caption{Relative impact (RI) of uncorrelated parameter noise on the classification test error over SEMEION and MNIST256 data, using the generic 8-qubit model circuit displayed in Figure \ref{Fig:maxent_circuit} }
\label{Tbl:parameter_noise}
\end{table}

In our simulation experiments we have systematically evaluated the effects on the quality of the classification of 0.1\%, 1\% and 10\% random perturbations in (a) the circuit parameters and (b) the input data vectors. As expected due to the unitariness, the effect of input noise is not amplified by the classifier circuit and thus had proportionate impact on the percentage of misclassifications. Somewhat more surprisingly, random perturbations of the trained circuit parameters almost never had the worst case estimated impact on the classification error. The 0.1\% uncorrelated parameter noise in the majority of cases had no impact on the classification results.\\

Here we limit the noise impact discussion to the context of the ``SEMEION'' and ``MNIST256'' data sets of the benchmark sets displayed in Table \ref{Tbl:datasets}. Our observations are easier to calibrate in this context since both datasets are encoded with $8$ qubits and the same model circuit architecture with $33$ gates at depth $19$ (as shown in Figure \ref{Fig:maxent_circuit} ) is used in the classifier.\\

The 0.1\% parameter noise had no impact on classification in about 60\% of our test runs. The maximum relative drift of the test error has been 3.5\% (in one the of remaining runs), the mean drift has been 1\% with the standard deviation of approximately 1.47\%. The 1\% parameter noise had a more pronounced, albeit fairly robust impact, which was non-trivial in about 90\% of our test runs. The maximum relative change in the test error rate has been 17\%, the mean relative change has been 7\% with the standard deviation of approximately 6.67\%. This statistics is summarized in Table \ref{Tbl:parameter_noise}. Finally the 10\% parameter noise lead to significant loss of classification robustness (although still smaller that the worst case analysis suggests). The maximum relative change in the test error rate has been 192.3\%, the mean has been 60.2\% with the standard deviation of 55.8\%. Curiously, the minimum change has been zero in one of the ``SEMEION''-based runs. This suggests that 10\% perturbation of parameters has no stable amplification pattern, and the model should be best re-trained after such perturbation.\\

The practical takeaway from these observations is that the circuit-centric classifiers may work on small quantum computers even in the absence of strong quantum error correction.

\section{Simulations and benchmarking}\label{Sec:sim}

To demonstrate that the circuit-centric quantum classifier works well in comparison with classical machine learning methods we present some simulations. The circuit-centric classifier was implemented on a classical computer using the F\# programming language.

\subsubsection{Datasets}

\begin{table*}[t]
\def\arraystretch{1.5}
\begin{tabular}{p{2.2cm} p{1.7cm} R{1.7cm} R{1.7cm} R{1.7cm}  p{0.7cm} p{7cm} }
 \toprule
ID & domain  & \# features & \# classes & \# samples & & preprocessing \\
 \hline
%BABI & reasoning & ?  & ? && \\
CANCER &decision&$32$ & $2$ & $569$ & &none \\
SONAR &decision& $60$ & $2$&$208$ & &padding with noninformative features \\
WINE &decision & $13$ & $3$ &$178$ & &padding with noninformative features \\
SEMEION &OCR\footnote{$Optical\,Character\,Recognition$}& $256$& $10$ &$1593$&& padding with noninformative features\\
MNIST256 &OCR& $256$&$10$& $2766$ & &coarse-graining and deskewing\\
%ARTIF & synthetic& $N$ & $2$ & $M$ & TODO:FILL OR KILL \\
\botrule
\end{tabular}
\caption{Benchmark datasets and preprocessing. }
\label{Tbl:datasets}
\end{table*}

\begin{spacing}{.7}
\begin{table*}[t]
\def\arraystretch{1.3}
\begin{tabular}{ p{2.cm} >{\raggedright\arraybackslash}p{3cm} >{\raggedright\arraybackslash}p{6.2cm} >{\raggedright\arraybackslash}p{6cm} }
\toprule
ID & model& fixed hyperparameters & variable hyperparameters  \\ \hline
QC  & Circuit-centric quantum classifier & entangling circuit architecture & dropout rate, number of blocks, range\\
PERC  & perceptron & - &  regularisation type \\
%SVMpoly  & Support vector machine & polynomial kernel $ \kappa(x,x') = (x^T x' +c)^d $ & regularisation strength of slack variables, \newline degree $d$, offset $c$ \\
%SVMrbf  & Support vector machine & radial basis function kernel $\kappa(x,x') = e^{-\gamma || x-x'||^2}$  & regularisation strength of slack variables, variance $\gamma$ \\
MLPlin  & neural network & dim. of hidden layer $=N$ & regularisation strength, optimiser, \newline initial learning rate \\
MLPshal  & neural network & dim. of hidden layer $=\lceil \log_2 N \rceil$ & regularisation strength, optimiser,  \newline initial learning rate\\
MLPdeep  & neural network & dim. of hidden layers $=\lceil \log_2 N \rceil$ & regularisation strength, optimiser,  \newline initial learning rate \\
SVMpoly1  & support vector machine & polynomial kernel with $d=1$, regularisation strength of slack variables $=1$, offset $c=1$ & - \\
SVMpoly2  & support vector machine & polynomial kernel with $d=2$, regularisation strength of slack variables $=1$, offset $c=1$ & - \\
\botrule
\end{tabular}
\caption{Benchmark models explained in the text and possible choices for hyperparameters. }
\label{Tbl:models}
\end{table*}
\end{spacing}

We select six standard benchmarking datasets (see Table \ref{Tbl:datasets}). The CANCER, SONAR, WINE and SEMEION data sets are taken from the well-known UCI repository of benchmark datasets for machine learning \footnote{The most recently accessed location of the UCI database is \url{https://archive.ics.uci.edu/ml/datasets.html} (more specifically, the CANCER data set is retrieved from \url{https://archive.ics.uci.edu/ml/datasets/Breast+Cancer+Wisconsin+\%28Diagnostic\%29}, the SEMEION from \url{https://archive.ics.uci.edu/ml/datasets/Semeion+Handwritten+Digit}, the SONAR from \url{https://archive.ics.uci.edu/ml/datasets/Connectionist+Bench+\%28Sonar\%2C+Mines+vs.+Rocks\%29} and the WINE from \url{https://archive.ics.uci.edu/ml/datasets/Wine})}. The MNIST data set is the official NIST Modified handwritten digit recognition data set. While CANCER and SONAR are binary classification exercises, the other data sets call for multi-class classification. Although the circuit-centric quantum classifier could be operated as a multi-class classifier, we limit ourselves to the case of binary classification discussed above and cast the multi-label tasks as a set of ``one-versus-all'' binary discrimination subtasks. For example, in the case of digit recognition the $i$th subtask would be to distinguish the digit ``$i$'' from all other digits. The results are an average taken over all subtasks. \\

The data is preprocessed for the needs of the quantum algorithm. The MNIST data vectors were course-grained into real-valued data vectors of dimension 256. With the exception of this data set, in all other cases the data vectors have been padded non-informatively so that their dimension after padding is the nearest power of $2$. After padding, each input vector was renormalized to unit norm. Thus each of the $N$-dimensional original data vectors would require $n=\lceil \log_2(N) \rceil$ qubits to encode in amplitude encoding. In our simulations we do not use feature maps which would multiply the input dimensions. However, we expect that the circuit-centric classifier will gain a lot of power from this strategy, which can be tested on real quantum devices without the same amount of overhead.

%In addition to the standard benchmarks, we generate an artificial dataset (ARTIF) in order to investigate the behaviour of the circuit-centric quantum classifier in more detail. \textcolor{red}{[TODO: How was the artificial data generated?]}

\subsubsection{Benchmark models}

For the circuit-centric quantum classifier (QC) we use the data-agnostic entangling circuit of $n$ qubits, which has been explained in Section \ref{Sec:architecture}. We use one, two or three entangling layers, which means that the circuit contains anywhere from $n+1$ to $6\, n +1$ single-qubit and two-qubit gates. Therefore the number of real trainable parameters varies from $3\,n+2$ to $18\,n+2$. For each dataset we selected the circuit architecture with the lowest training error while reducing overfitting.\\

Defining a fair, systematic comparison is not straight forward since there are many different models and training algorithms we could compare with, each being further specified by hyperparameters. Without the use of feature maps, our classifier has only limited flexibility, and benchmarking against state-of-the-art models such as convolutional neural networks will therefore not provide much insight. Instead, we choose to compare our model to six classical benchmark models (see Table \ref{Tbl:models}) that are selected for their mathematical structure which is related to the circuit-centric classifier. \\

Section \ref{Sec:nn} showed an interesting parallel to neural networks, which is why we take neural networks as one benchmark model family. From this family we choose $3$ different architectures shown in Figure \ref{Fig:mlp_architectures}. The MLPlin model has a linear hidden layer of the same dimension $N$ as the input layer and resembles the architecture of our QC model displayed in Figure \ref{Fig:model}. The MLPshal and MLPdeep models have hidden layers of size $\lceil \log_2 N \rceil$. The motivation is to compare the QC with an architecture that - after the input layer of size $N$ - gives rise to a polylogarithmic number of weights. \\

\begin{figure}[t]
\def\arraystretch{1.5}
\centering
\includegraphics[width=0.45\textwidth]{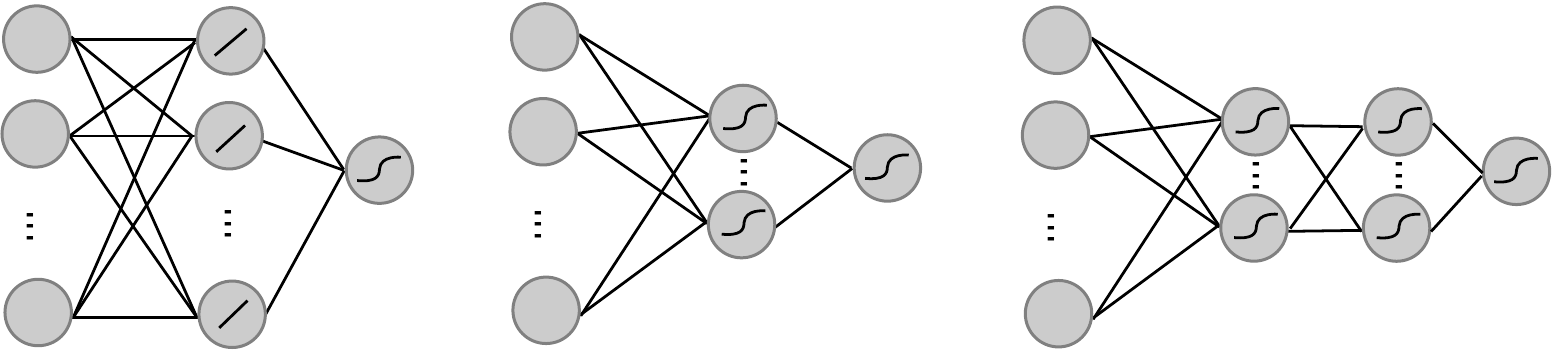} 
\caption{The three architectures of the benchmark artificial neural network models. Left is the MLPlin model, which has a hidden linear layer of the size of the input $N$ and a logistic output layer. The MLPshal model in the middle has a hidden layer of size $\lceil \log_2 N \rceil$ with nonlinear activations and a logistic output layer. The MLPdeep model model on the right adds a second nonlinear hidden layer.  }
\label{Fig:mlp_architectures}
\end{figure}

We use a support vector machine as a second benchmark model. Support vector machines are similar to the circuit-centric classifier since we can also think of them as a linear classifier in a feature space that is defined by a \textit{kernel} $\kappa$ \cite{scholkopf02}. We mentioned in Section \ref{Sec:stateprep} that amplitude encoding can be supplemented by a feature map which is very similar to that of a support vector machine with a polynomial kernel,
\[\kappa(x,x') = (x^T x' +c)^d.\]
The offset $c$ can be loosely compared to the effect of padding. The degree $d$ of the kernel is not one-to-one comparable to the degree of the polynomial feature map in amplitude encoding, since our QC model effectively contains an `extra' nonlinearity which derives from the measurement process (and is therefore not precisely a linear model in feature space). This can be seen in Figure \ref{Fig:2d} where we compare the decision boundaries of an SVM with polynomial degree $d=1,2$ with the QC model and a feature map of degree $d=1,2$. To counterbalance the potential advantage of the QC, we consider two support vector machines, one with $d=1$ (SVMpoly1) and one with $d=2$ (SVMpoly2). The QC model does not use any feature maps. Finally, we add a perceptron (PERC) to the list of benchmark models to get an impression about the linear separability of the datasets.\\

\begin{figure}[t]
\centering
\includegraphics[width=0.47\textwidth]{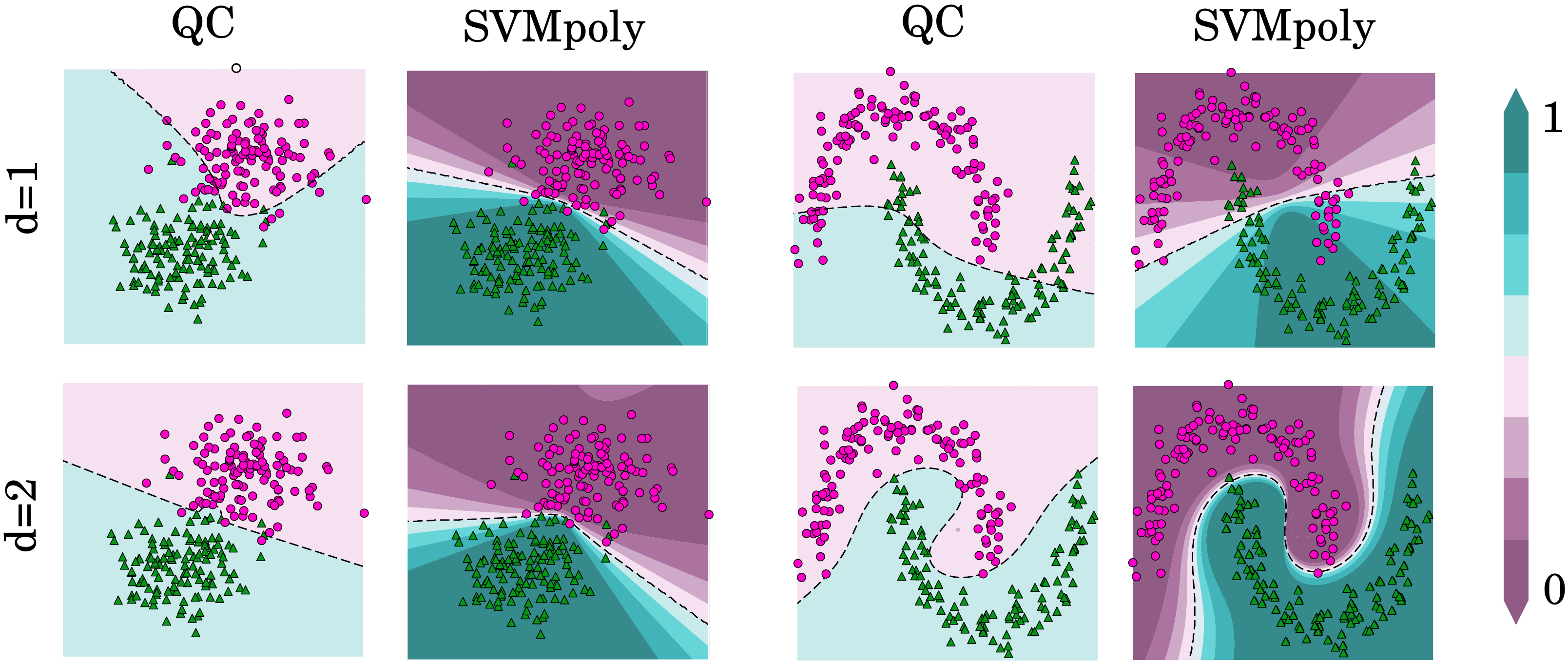}
\caption{Comparison of the decision boundary for the circuit-centric quantum classifier (QC) and a support vector machine  with polynomial kernel (SVMpoly). The $2$-dimensional data gets embedded into a $4$-dimensional feature space (we padded with $2$ non-informative features), where it is classified by the two models. The colour scale indicates the probability that the model predicts class ``green circles''.  For the QC model, the parameter $d$ corresponds to the degree of the polynomial feature map in amplitude encoding (see Section \ref{Sec:stateprep}). For the SVMpoly, $d$ is the degree of the polynomial kernel. One can see that for $d=1$ the QC model is slightly more flexible than the SVMpoly.  }
\label{Fig:2d}
\end{figure}

Since one of our goals is to build a model with a small parameter space, we compare the number of trainable parameters for each model in Table \ref{Tbl:paramcount}. For the MLP and PERC models these parameters are the weights between units, and their number is defined by the dimension of inputs as well as the architecture of the network. For the SVM models we consider the number of input data points, since in their dual formulation these models start with assigning a weight to each input, after which they reduce the training set to a few support vectors used for inference.

\begin{table}[t]
\def\arraystretch{1.5}
\begin{tabular}{p{1.2cm} R{1.cm} R{0.9cm} R{1.2cm} R{1.2cm}  R{1.2cm} R{0.9cm} }%R{1.7cm}}
 \toprule
\footnotesize ID & \footnotesize QC & \footnotesize PERC & \footnotesize  MLPlin &\footnotesize MLPshal & \footnotesize MLPdeep  & \footnotesize SVM\\%& SVMrbf  \\
 \hline
\footnotesize CANCER   & $79$ &$32$ & $1056$ & $165$& $190$&$208$\\% &$208$\\
\footnotesize SONAR    & $60$ &$60$ & $1952$ &$305$&$330$&$569$ \\% &$569$\\
\footnotesize WINE     & $28$&$13$ &$272$ &$51$&$60$ &$178$\\%&$178$\\
\footnotesize SEMEION  & $100$ &$256$ &$65792$ &$2056$ &$2120$ &$1593$\\%&$1593$\\
%SEMEION  & $100$ &$256$ &&&&$1593$&$1593$\\
\footnotesize MNIST256 & $124$&$256$ &65792&2056&2120&$800$\\%&$1570$\\
%MNIST256 & $100$ &$256$ & $65792$ &$2056$& $2120$ &$2766$&$2766$\\
\botrule
\end{tabular}
\caption{Number of parameters each model has to train for the different benchmark datasets. The circuit-centric quantum classifier QC has a logarithmic growth in the number of parameters with the input dimension $N$ and data set size $M$, while all other models show a linear or polynomial parameter growth in either $N$ or $M$.}
\label{Tbl:paramcount}
\end{table}

\subsubsection{Results}

\begin{table*}[t]
\def\arraystretch{1.5}
\begin{tabular}{ p{2cm} R{2.5cm} R{2.5cm} R{2.5cm} R{2.5cm} R{2.5cm}}
\toprule
 & CANCER &  SONAR  & WINE$^*$ & SEMEION$^*$  & MNIST256$^*$ \\  \hline
QC  &  $0.022/\mathbf{0.058}$ & $0.000/0.195$ & $0.000/0.028$ & $0.031/0.031$ & $0.031/0.033$ \\ %[alexeib] 3/19/2018; 1st correction for cross-vld mean;
PERC  &$0.128/0.137$ &$0.283/0.315$& $0.067/0.134$& $0.022/0.038$& $0.065/0.066$\\
%SVMpoly  &$0.030/0.059$ &$0.000/0.187$& $0.018/0.043$& $0.000/0.028$&$0.003/0.023$\\
%SVMrbf  &$0.008/0.060$&$0.000/\mathbf{0.130}$ &$0.000/\mathbf{0.021}$&$0.000/\mathbf{0.012}$&$0.000/\mathbf{0.006}$\\
MLPlin  &$0.060/0.075$ & $0.117/0.263$& $0.001/\mathbf{0.039}$&$0.001/0.025$& $0.038/0.041$\\
MLPshal &$0.064/0.077$& $0.010/\mathbf{0.174}$ &$0.029/0.063$ &$0.002/\mathbf{0.024}$&$0.011/\mathbf{0.018}$\\
MLPdeep  &$0.056/0.076$&$0.001/\mathbf{0.174}$&$0.010/0.063$&$0.001/0.026$&$0.014/0.021$\\
SVMpoly1 &$0.373/0.367$&$0.452/0.477$&$0.430/0.466$& $0.101/0.100$ &$0.092/0.092$\\
SVMpoly2 &$0.169/0.169$&$0.334/0.383$&$0.090/0.099$ & $0.100/0.101$ &$0.091/0.092$\\\hline
Average & $0.125/0.136$ & $0.171/0.283$ & $0.090/0.137$ & $0.037/0.049$ &$0.040/0.043$\\  \botrule
\end{tabular}
\caption{Results of the benchmarking experiments. The cells are of the format 	``training error/validation error''. The variance between the $50$ repetition for each experiment was of the order of $0.01-0.001$ for the training and test error. $^*$For multilabel classification problems with $d$ labels, the average of all $d$ one-versus-all problems train and test errors were taken.  }
\label{Tbl:results}
\end{table*}

For every benchmarking test (except from the QC model on SEMEION and MNIST256) we use $5$-fold crossvalidation with $10$ repetitions each. This means that the results are an average of $50$ repetitions of training the model and calculating the training and test error. Due to the significant cost of quantum circuit simulations, for the QC experiments on the SEMEION and MNIST256 datasets only one repetition of the $5$-fold cross-validation was carried out. The results are summarized in Table \ref{Tbl:results}. \\

As we can read from the non-zero training error of the PERC model, none of the datasets is linearly separable. One finds that the QC model performs significantly better than the SVMpoly1 and SVMpoly2 models across all data sets. In further simulations we verified that for support vector machines with polynomial kernel, degrees of $d=6$ to $d=8$ return the best results on the datasets, which are also better than those of the MLP models.\\

Although showing slightly worse test errors than the MLPshal and MLPdeep (and with mixed success compared to the MLPlin), the QC performs comparable with the MLP models that train a lot more parameters. For the SONAR and WINE dataset we find that the QC model overfits the training data. This is even worse without using the dropout qubit regularization technique explained above. The observation is interesting, since the QC model is `slimmer' than the MLP and SVM models in terms of its parameter count. An open question is therefore how to introduce other means of regularisation.

\section{Conclusion}\label{Sec:concl}

We have developed a machine learning design which is both quantum inspired and implementable by near-term quantum technology. The key building block of this design is a unitary model circuit with relatively few trainable parameters that
assumes amplitude encoding of the data vectors and uses systematically the entangling properties of quantum circuits as a resource for capturing correlations in the data. After state preparation, the prediction of the model is computed by applying only a small number of one- and two-qubit gates quantum gates. At the same time, simulating these gates on a classical computer requires a of elementary operations that scales with the number of features. We are aware of prior designs of unitary neural nets available in literature (cf. \cite{arjovsky15, jing16} and related work). In these designs the number of learnable parameters is proportional to the number of data features, whereas the size of our model circuit therefore scales with the number of qubits and thus allows, qubit-wise, for exponentially fewer learnable parameters than what traditional methods would use. We have also shown in some preliminary simulations that the design can indeed achieve results close to off-the-shelf methods that have comparable limitations but considerably more tunable parameters. \\

The quantum classifiers based on model circuits that we have explored so far belong to a class of weakly nonlinear classifiers. The main source of nonlinearity in our classifiers stems from the concluding measurement and thresholding on the probabilities of the measurement outcomes. These probabilities are roughly quadratic in the amplitudes of the final post-circuit state. Since the overall effect of the model circuit on the amplitudes is linear reversible, one concludes that the models we have experimented with can capture (amplitude-wise) quadratic separation of classes in the original feature space. \\

There is a potential for tracing class separation boundaries of higher polynomial degrees by encoding several copies of a classical data vector in disjoint quantum registers and then having the model circuit work on the tensor power of the data vector. Of course, this setup requires vastly more computational resources to simulate the quantum circuit, and we leave such experiments for the future. We furthermore expect that the most beneficial application of our quantum classifiers is to quantum data. One can conceive entangling a quantum system with a classifier circuit and use the latter to discriminate between various states of the quantum system.\\

This work contributes to the growing literature on variational circuits for machine learning applications in proposing a specific circuit architecture and parametrisation, a dropout regularisation technique, a hybrid training method, as well as a graphical interpretation of quantum operations in the language of neural networks. However, an overwhelming number of questions is still largely unexplored. For example, we noticed that our slim circuit design still suffers from overfitting. Also, full-fledged numerical benchmarks on larger datasets are needed to systematically analyse the effect of logarithmically few parameters with growing input spaces. The power of feature maps to introduce nonlinearities is another open question. Further numerical as well as theoretical studies are therefore crucial to understand the convergence properties and representational power of circuit-centric models for classification.

%[TODO:REDO] We have evaluated our proposed models and training strategies on a usual slate of UCI machine learning benchmarks, available from the UCI database of classical machine learning data. In a nutshell, the classification quality of our classifiers has been better than that of purely linear neural nets, on a par with polynomial-kernel SVMs, slightly worse than that of traditional multilayer perceptrons and yielding to both SVMs with RBF kernels and Deep Neural Nets. One must note, however, that in most cases reasonable classification quality has been achieved by our models with the parameter count that has been dramatically smaller than the parameter count in the off-the-shelf comparison models. An important direction of future research would be to find a framework for building competitive classifiers that would use our circuit-centric models as building blocks.

\acknowledgements
The Authors are grateful to Jeongwan Haah for insightful remarks and wish to thank Martin Roetteler for numerous useful discussions of the ongoing research and early drafts of this paper. MS wishes to thank the entire QuArC group at Microsoft Research, Redmond, whose collective help during her research internship had been generous, professional and very welcoming.

\bibliography{litArchive}
\bibliographystyle{unsrt}

\end{document}